%% file: arXiv.tex
\documentclass[a4paper,11pt,english,final]{article}
\pdfoutput=1 

\usepackage[utf8]{inputenc}
\usepackage{fullpage}
\usepackage[refpage,prefix]{nomencl}
\newcommand*{\nom}[2]{#1\nomenclature{#1}{#2}}
\usepackage[displaymath,textmath,sections,graphics,floats]{preview}
\makenomenclature

\usepackage[english]{babel}
\usepackage{verbatim}

\usepackage{tikz}
\usepackage{amssymb}
\usepackage{mathtools}
\usepackage{bm}
\usepackage{bbm}

\usepackage{algorithm}
\usepackage{algorithmic}
\usepackage{enumitem}
\usepackage{float}

\usepackage{titling}

\newcommand{\pvp}{\vec{p}{\kern 0.45mm}'}

\DeclarePairedDelimiter\bra{\langle}{\rvert}
\DeclarePairedDelimiter\ket{\lvert}{\rangle}
\DeclarePairedDelimiterX\braket[2]{\langle}{\rangle}{#1 \delimsize\vert #2}
\newcommand{\underflow}[2]{\underset{\kern-60mm \overbrace{#1} \kern-60mm}{#2}}

\def\Ind(#1){{{\tt Ind}({#1})}}
\def\Id{\mathrm{Id}}

\def\Tr{\mathrm{Tr}}
\def\im{\mathrm{im}}

\newcommand{\QMAo}{\textsf{QMA$_1$}}
\newcommand{\BQP}{\textsf{BQP}}
\newcommand{\NP}{\textsf{NP}}
\newcommand{\SharpP}{\textsf{\# P}}

\long\def\ignore#1{}

\newtheorem{theorem}{Theorem}
\newtheorem{corollary}[theorem]{Corollary}
\newtheorem{lemma}[theorem]{Lemma}
\newtheorem{prop}[theorem]{Proposition}
\newtheorem{definition}[theorem]{Definition}
\newtheorem{claim}[theorem]{Claim}
\newtheorem{remark}[theorem]{Remark}

\newenvironment{proof}
{\noindent {\bf Proof. }}
{{\hfill $\Box$}\\	\smallskip}
\usepackage{ifdraft}
\ifdraft{\newcommand{\authnote}[3]{{\color{#3} {\bf  #1:} #2}}}{\newcommand{\authnote}[3]{}}

\newcommand{\onote}[1]{\authnote{Or}{#1}{green}}
\newcommand{\anote}[1]{\authnote{Andras}{#1}{red}}

\usepackage[breaklinks,final]{hyperref}

\usepackage{ifpdf} 
\ifpdf
	\input{myBiblatex.tex}
\else
	\usepackage[quadpoints=false]{hypdvips}
\fi 

\title{
	On preparing ground states of gapped Hamiltonians:\\ 
	An efficient Quantum Lovász Local Lemma}
\author{András Gilyén\thanks{QuSoft, CWI and University of Amsterdam, the Netherlands. \texttt{gilyen@cwi.nl} }
	\and
	Or Sattath\thanks{The Hebrew University and MIT. \texttt{sattath@cs.huji.ac.il}}%
}
\date{\vspace{-12mm}}

\begin{document}
	
	\maketitle
	
	\providecommand{\trnorm}[1]{\left\lVert#1\right\rVert_1}
	\providecommand{\infnorm}[1]{\left\lVert#1\right\rVert_{\infty}}
	\providecommand{\spnorm}[1]{\left\lVert#1\right\rVert}
	\providecommand{\tr}[1]{\Tr\left[#1\right]}
	\providecommand{\sgn}[1]{\mathrm{sgn}\left(#1\right)}
	\providecommand{\rank}[1]{\mathrm{rank}\!\left(#1\right)}

	\begin{abstract}
		A frustration-free local Hamiltonian has the property that its ground state minimises the energy of all local terms simultaneously. In general, even deciding whether a Hamiltonian is frustration-free is a hard task, as it is closely related to the $\QMAo$-complete quantum satisfiability problem (QSAT) -- the quantum analogue of SAT, which is the archetypal $\NP$-complete problem in classical computer science. This connection shows that the frustration-free property is not only relevant to physics but also to computer science. 

		The Quantum Lovász Local Lemma (QLLL) provides a sufficient condition for frustration-freeness. A natural question is whether there is an efficient way to prepare a frustration-free state under the conditions of the QLLL. Previous results showed that the answer is positive if all local terms commute.    
		
		In this work we improve on the previous constructive results by designing an algorithm that works efficiently for non-commuting terms as well, assuming that the system is ``uniformly" gapped, by which we mean that the system \emph{and all its subsystems} have an inverse polynomial energy gap. Also, our analysis works under the most general condition for the QLLL, known as Shearer's bound. Similarly to the previous results, our algorithm has the charming feature that it uses only local measurement operations corresponding to the local Hamiltonian terms.
	\end{abstract}
	\onote{To do:		
		-disipation - verstrate 
		
		For next version:
        - Add a simple proof using the compression argument.
        - Marriott-Watrous kind of proof for the exact unitary version, instead of the dimension counting argument.

        - Or: I would still like to understand better the algorithm which picks a flaw (either randomly, or by the natural order), measure it, and if it is unhappy, resample it, and repeat this process indefinitely. Isn't it true that the probability of seeing a log with m failures is $p^m$? It's unclear what would be the generalisation of the key lemma: what could we say about $\rho_L$? Andras: Do you have any insight?
        
 - Is there an "adiabatic algorithm" variant? Perhaps it could explain the running time: that the running time of the adiabatic algorithm is indeed a function of the uniform gap. 
	}
	
	\section{Introduction}
	
	\paragraph{Frustration-free Hamiltonians and quantum satisfiability.} 
    Most physical systems and models are described by a local Hamiltonian $H=\sum_i{H_i}$ where each k-local term $H_i$ acts non-trivially only on at most $k$ of its subsystems. 
    Such a Hamiltonian is called \emph{\nom{frustration-free}{A local Hamiltonian in which the ground state is also the ground state of each of the local terms}} if its ground state is also the ground state of each of the local terms $H_i$. Frustration-free Hamiltonians appear in various areas, for example: quantum error correcting codes~\cite{Gottesman96class}, parent Hamitlonians for PEPS (a 2-D generalisation of matrix-product-states)~\cite{perez08PEPS}, and various models in many-body quantum physics.  
    
    An equivalent way to ask whether a Hamiltonian $H$ is frustration-free is whether $H'=\sum_i \Pi_i$ is frustration-free, where $\Pi_i$ is the projector on the excited states of $H_i$. The quantum satisfiability problem (QSAT)~\footnote{For technical  reasons, that would not be relevant for this work, there is a promise that if $H'$ is not frustration-free, the minimal energy of $H'$ is at least inverse polynomial in the number of qubits.} is to determine whether $H'$ in the above form is frustration-free. QSAT is $\QMAo$-Complete~\cite{bravyi2011efficient}, and therefore intractable in general even for quantum computers (unless $\BQP = \QMAo$). Finding the ground state of frustration-free Hamiltonians -- the challenge we tackle in this work -- is, in general, an even harder task. 
    
	
	
	\paragraph{The Classical and Quantum Lovász Local Lemma.} 
	We would like to understand the QSAT problem, so it is natural to first look at the classical SAT and the techniques that were useful in studying it. A ``local" version of SAT is called $k$-SAT which asks whether a Boolean formula of the following form can be satisfied: $\bigwedge_{i\in [m]}C_i$, where each $C_i$ is a {\em clause } containing the or ($\bigvee$) function of $k$ Boolean variables or their negation. 
	
	A natural question is, when can we be sure that a satisfying assignment exists?
	Since each $k$-SAT constraint excludes a $p=2^{-k}$ fraction of assignments, $pm<1$ is a sufficient condition (by the union bound). If we have the additional information that none of the constraints share variables, then it is clearly satisfiable. What can we say in the intermediate regime, where each constraint shares variables with at most $d$ constraints (including itself)? The (symmetric) Lovász Local Lemma~\cite{Erdos73,AlonS92,Szegedy13}, applied to this setting, implies that $p d e \leq 1$ is a sufficient condition for satisfiability. Shearer generalised the Lovász Local Lemma and showed the weakest possible sufficient condition in this framework~\cite{Shearer}. 
	
	How hard is it to find such a satisfying assignment? A series of works~\cite{beck91a,MoserOrig,MoserTardos,KolipakaSzegedy} have culminated in an efficient constructive algorithm, even under Shearer's condition.  
	
	It is natural to ask the analogous questions in the quantum setting, where the Boolean variables are replaced by qubits and the clauses by rank-1 $k$-local projectors. The resemblance between k-SAT clauses and rank-1 projector is the following: a k-SAT clause excludes one out of the $2^k$ possible configurations of the relevant variables, whereas a rank-1 k-local projector excludes one dimension out of the $2^k$ relevant dimensions. In the quantum setting it makes sense to generalise and consider rank-$r$ projectors. So given a set of $k$-local rank-$r$ projectors acting on $n$ qubits, under what conditions can we guarantee that the system is frustration-free? A ``dimension-counting" argument can be used to show that the Lovász condition ($2^{-k}r d e \leq 1$)~\cite{LovAmb} is indeed sufficient, as is Shearer's condition~\cite{SattathLatice}.
	
	Is there an algorithm which efficiently prepares a ground state under these conditions? In the past, such constructions have been achieved only for commuting Hamiltonians, i.e. $[\Pi_i,\Pi_j]=0$ for all $i,j$. Commuting Hamiltonians are somewhat ``half-way" between classical and quantum. For example, the commuting 2-local Hamiltonian problem is in (the purely classical class) $\NP$ for qudits of all dimensions~\cite{bravyi05commutative}, whereas $2$-local QSAT is $\QMAo$-Complete if the dimension of the qudits is large enough~\cite{aharonov09power}. Yet commuting Hamiltonians, such as the toric code, can have the striking quantum property of topological order~\cite{KitaevToric}. Also, the commuting existential QLLL is a direct corollary of the (classical) existential LLL.
	
	The analysis of the previous algorithms~\cite{SchwarzInfo,SattathSymm}, that worked only in the commuting case, used a compression argument, while requiring the symmetric Lovász condition. There was another attempt~\cite{CubittBackwards} to prove a stronger version inspired by the classical ``backward-looking" analysis of Moser and Tardos~\cite{MoserTardos,KolipakaSzegedy}. However, as noted on the arXiv~\cite{CubittBackwards}, there is an unresolved gap in the proof of the main result, due to an issue of non-commutativity of subsequent resamplings. (The resampling operation is to replace some (qu)bits with uniformly random (qu)bits.) In fact, on the classical side, Kolmogorov~\cite{KolmogorovComm} argued that this kind of ``backward-looking" analysis requires a sort of commutation property for resamplings, which, as we show in  Appendix~\ref{apx:commExampla}, fails in general for the quantum case, even when the projectors commute. 
	
	Recently Harvey and Vondrák~\cite{HarveyVondrak15} introduced a classical ``forward-looking" analysis technique which gives slightly worse bounds on the expected number of resamplings, but works in a more general framework and requires only Shearer's condition. This framework is quite flexible and allowed us to transform the results to the quantum setting as well as addressing the non-commuting case.
	
	\paragraph{The gap constraint.}	
	The gap of a Hamiltonian -- the energy difference between its (distinct) lowest energy levels --  denoted $\Delta(H)$, plays an important role both in physics and computer science, particularly in Hamiltonian complexity theory, see, e.g.,  ~\cite{hastings07area,farhi00quantum,aharonov08adiabatic,cubitt15undecidability}. 
	Suppose $H=\sum_{i\in[m]}\Pi_i$, then we call the uniform gap of $H$: (For $H'=0$ we define $\Delta(H')=\infty$.)
	\begin{equation} \label{eq:def_gamma}
	\gamma(H) := \min_{S \subseteq [m]} \Delta\left(\sum_{i \in S}  \Pi_i\right).
	\end{equation} 
	This notion of uniform gap plays an important role in another recent state preparation algorithm \cite{GeMolnar}.
	The running time of our algorithm has inverse quadratic dependence on the uniform gap. 

	\paragraph{Our contribution.}
	In this work we prove a constructive Quantum Lovász Local Lemma for \emph{non-commuting} projectors.
	We will only consider systems of $n$ qubits, but all the results generalise trivially to qudits. 	
	Our improvements are due to two main new ingredients. 
	
	The first ingredient is the adaptation of the ``forward-looking" analysis technique~\cite{HarveyVondrak15} to the quantum setting.
	This analysis technique enables us to go beyond the symmetric Lovász condition, and prove efficiency under Shearer's (weaker) condition. 
	Technically, for this to work, we need to bound not only the probabilities of some bad events, but show that the quantum state is bounded above by a uniformly mixed state in the satisfying subspace (see Definition~\ref{def:progressive_channel},  and Lemma~\ref{lemma:NonComm}). 
       
    The second technique is the use of weak measurements coupled with a quantum Zeno effect. We consider this to be the main contribution, and devote the next two subsections to explain its purpose and behaviour in our context. 
    
    In the next arXiv version of this manuscript, we will also present an alternative algorithm for the non-commuting case working under the symmetric Lovász condition. The analysis of that algorithm is somewhat simpler and also gives improved bounds on the  running time, as it is based on the entropy compression argument of~\cite{SchwarzInfo,SattathSymm}. 
	
	\paragraph{Loop invariant.} The classical analysis techniques we mentioned differ in how they prove bounds on the runtime of the constructive algorithms. However the basic idea for proving correctness is the same for the corresponding algorithms. The initial state is a uniformly random state. The algorithm starts without knowing which constraints are satisfied, and its goal is to enlarge that list of satisfied constraints. It checks whether a constraint is satisfied: if it is satisfied, it is added to the list. Otherwise, it uses a recovery procedure: it resamples all the variables involved in that constraint, and removes all the constraints from the list that may have been affected by the resampling. The algorithm terminates when the list is complete, and thus a satisfying assignment is found.
	
	This correctness proof works in the classical case and carries through for the commuting quantum case, but fails in the non-commuting case. The main problem is with the loop invariant: if a set of constraints $S$ was satisfied, and then another constraint projector $\Pi_f$ was checked (i.e., measured) and found to be satisfied, the post-measurement state does not necessarily satisfy all the constraints in $S$ (to be more specific, the constraints in $S$ that share qubits with $\Pi_f$), because of the collapse caused by the measurement.  
	
	\paragraph{Weak measurements.} We get around the  difficulty of maintaining the loop invariant by using a kind of \emph{weak measurement} and the quantum Zeno effect.
	This approach is somewhat similar to the ideas described in \cite{PlatoCave}.
	Instead of measuring whether $\Pi_f$ is satisfied, we repeat the following many times: we perform a weak measurement (as explained below) to find whether $\Pi_f$ is satisfied. If it is not, we just apply the usual resampling step. If it is satisfied, we (strongly) measure  whether all constraints in $S$ are simultaneously satisfied. If they are not then we abort, and repeat otherwise. When the loop ends, we measure whether $S \cup \{\Pi_f\}$ is simultaneously satisfied, and abort if they are not. 
	
	By tuning the weak measurement parameter and the number of repetitions, we can control and reduce the probability of aborting in this procedure. Therefore, the two probable outcomes are that we either end up with adding $\Pi_f$ to the set of satisfied clauses, or  we use the same recovery procedure, that worked in previous cases, and still works in the non-commuting case.
	
	One may wonder: if the probability of abort is kept small, are these measurements really necessary? Yes -- similar to the ``hot pot never boils" phenomenon, and the quantum Zeno effect, even though the outcome of the measurement is known with very high probability in advance, the measurement changes the overall state dramatically when applied frequently. 
	
	Now we explain what we mean by a weak measurement, and how it can be combined with the quantum Zeno effect. Consider the two-outcome measurement $\{\Pi_f, \Id-\Pi_f\}$.	We can implement a weak measurement on $\ket{\psi}$ with intensity parameter $\theta$ using an ancilla qubit and a $\Pi_f$-controlled rotation 
	
	\begin{equation}\label{eq:cRot}
	\Pi_f^\theta=\Pi_f\otimes R^\theta+(\Id-\Pi_f)\otimes \Id \text{ , where } 
	R^\theta=\left(\begin{array}{cc}
	\sqrt{1-\theta} & -\sqrt{\theta}  \\ 
	\sqrt{\theta} & \sqrt{1-\theta}
	\end{array}\right) .
	\end{equation}
	We simply apply $\Pi_f^\theta$ on $\ket{\psi}\otimes\ket{0}$ and do a projective measurement on the ancilla qubit.
	Let us denote by $\ket{\psi_1}=\sqrt{\theta}\Pi_f\ket{\psi}$ the (unnormalised) state corresponding to outcome $1$. So $\ket{\psi_1}\propto\Pi_f\ket{\psi}$ just as we expect from a projective (strong) measurement.
	Similarly let $\ket{\psi_0}=(\Id-\Pi_f)\ket{\psi}+ \sqrt{1-\theta}\Pi_f\ket{\psi}\approx\ket{\psi}-(\theta/2)\Pi_f\ket{\psi}$ denote the (unnormalised) state corresponding to outcome $0$.
	
	Suppose $\ket{\psi}=\Pi_V \ket{\psi}$ for some other orthogonal projector $\Pi_V$. The probability of measuring $0$ on the ancilla qubit and finding the state outside the support of $\Pi_V$ has probability 
	$\spnorm{\left(\Id-\Pi_V\right)\ket{\psi_0}}^2\approx \spnorm{\left(\Id-\Pi_V\right)\left(\ket{\psi}\!-\!(\theta/2)\Pi_f\ket{\psi}\right)}^2
	=\spnorm{\left(\Id-\Pi_V\right)(\theta/2)\Pi_f\ket{\psi}}^2\leq \theta^2\spnorm{\Pi_f\ket{\psi}}^2$. 
	
	On the other hand we can argue, that conditioned on never seeing a $1$ outcome, $\sim 1/\theta$ repetitions of the weak and strong measurements dissipate the overlap of $\ket{\psi}$ with $\Pi_f$ almost completely, while leaving the part lying in $\im\left(\Pi_V\right)\cap\ker\left(\Pi_f\right)$ undisturbed.	
	Observe that the overall probability of error, i.e., moving out of the support of $\Pi_V$, is at most $\sim \theta$.
	Analogously to the quantum Zeno effect, setting $\theta$ small enough we can go below any desired error probability. This argument lies at the heart of the proof. 
	
	In some sense our error bound is even stronger than in the usual quantum Zeno effect: the probability of moving out of the support of $\Pi_V$ is proportional to $\left\lVert\Pi_f\ket{\psi}\right\rVert^2$, so the smaller the overlap with $\Pi_f$ gets, the smaller the error probability becomes. Because of this we can show that the overall probability of this error is bounded by $\theta$ independent of the number of repetitions.
	
	\paragraph{The algorithm.} The above argument shows that by using weak enough measurements in the following algorithm, we can ensure a high probability of termination with ``SUCCESS'' if we can bound the number of repetitions of the main while loop. This bound is obtained by using the techniques of \cite{HarveyVondrak15}. The following algorithm maintains a list of already checked projectors $C$. We denote by $\Pi^C$ the orthogonal projection having kernel equal to the intersection of the kernels of all the projectors in $C$.  (Note that $C$ and thus $\Pi^C$ changes during the algorithm.)
	\begin{algorithm}[H]
		\caption{High-level description of the main algorithm}\label{alg:symm}
		\begin{algorithmic}[1]
			\STATE {\bf input} constraints $\{\Pi_f\}_{f\in F}$
			\STATE set all qubits to the maximally mixed state, and mark all constraints as unchecked.
			\STATE {\bf while} there is a $\Pi_f$ which is unchecked {\bf do}
			\STATE ~~~ {\bf repeat} $T$ times {\bf do}
			\STATE ~~~~~~ measure $\Pi_f$ weakly
			\STATE ~~~~~~  {\bf if} the measurement found $\Pi_f$ violated {\bf then}
			\STATE ~~~~~~~~~  resample all qubits of $\Pi_f$, i.e., replace them by uniformly random qubits
			\STATE ~~~~~~~~~  mark all constraints $\Pi_{f'}$ that share qubits with $\Pi_{f}$ as unchecked
			\STATE ~~~~~~~~~  go back to the beginning of the main while loop
			\STATE ~~~~~~  {\bf end if}
			\STATE ~~~~~~ measure $\Pi^C$ {\bf if} it is violated 
			{\bf then terminate with } 	$\!\!$``ERR: MEASURE WEAKER"\kern-10mm
			\STATE ~~~ {\bf end repeat}				
			\STATE ~~~ mark $\Pi_f$ as checked
			\STATE ~~~ measure $\Pi^C$ {\bf if} it is violated 
			{\bf then terminate with } 	$\!\!$``ERR: USE LARGER T"				
			\STATE {\bf end while}
			\STATE {\bf terminate with} ``SUCCESS" 		
		\end{algorithmic}
	\end{algorithm}
	
	We develop an approximate version of the above algorithm which overcomes the need for using non-local operators $\Pi^C$.
	The only quantum operations that our approximate algorithm uses are (weak and strong) measurements of the projectors and resampling of qubits.
    As we show in Corollary~\ref{cor:SLCRuntime}, under the symmetric Lovász condition, the runtime of our quantum algorithm is
	$\tilde{\mathcal{O}}\left(\frac{m^3\cdot n^2}{\gamma^2}\cdot
	\ln\left(\frac{1}{\delta}\right)\cdot\ln\left(\frac{1}{\epsilon}\right)\right)$,\footnote{By $\tilde{\mathcal{O}}(t)$ we mean $\mathcal{O}(t\cdot \text{poly(log}(t)))$, moreover here the poly is actually quadratic. } where $n$ is the number of qubits, $m$ is the number of projectors, $\gamma$ is the uniform gap (see \eqref{eq:def_gamma}),
	$\delta$ is the desired maximum trace distance from a density operator which is supported on the ground space,
	and $\epsilon$ is the desired upper bound on the probability of termination with ``ERROR".
    The exact formula for the runtime bound we prove in the general Shearer case is more complicated, but it is easy to compare to the classical case. 
	Let $R_c$ be the upper bound of \cite{KolipakaSzegedy} on the expected number of resamplings of the classical Moser-Tardos algorithm. 
	Then, our quantum algorithm has runtime 
	$\tilde{\mathcal{O}}\left(\frac{R_c^2n^2m^2}{\gamma^2}\log\left(\frac{1}{\delta}\right)\cdot\log\left(\frac{1}{\epsilon}\right)\right)$, see Corollary~\ref{cor:boosted}.  \anote{It can be improved by $n^2$ using the compression argument, and I thought it can be improved by $m$ using phase-estimation, but now I do not think so. For phase estimation Hamiltonian simulation is needed. But our Hamiltonian has sparsity $\sim m$ and also norm $\sim m$, yielding an $\sim m^2$ cost. The corresponding cost coming from weak measurements is $\sim m$ because the strong measurement requiring a lot of local measurements, and there is a $\sim R$ factor because of the weakness parameter. This is actually smaller for the SLC case, but might be a bit higher in other cases.}
	
	\paragraph{New existential proof} Our work does not require any of the previous existential proofs, and therefore provides an alternative proof for the results in~\cite{LovAmb} and~\cite{SattathLatice}, see Corrolary~\ref{cor:existential}.
    
    \paragraph{Structure of the paper}
	In Section~\ref{sec:Definitions} we list the most important definitions and the notations that we use throughout the paper.

	In Section~\ref{sec:forward} we describe our generalisation of the projective measurement step from the commuting \cite{SchwarzInfo,SattathSymm} to the non-commuting case in terms of quantum channels.
	In \ref{subsec:progressMeas} we describe our loop-invariants which define progress in terms of subspaces. 
	In \ref{subsec:idealChannel} we define our ideal quantum channel introduced for the non-commuting case.
	Since we cannot implement this ideal operation efficiently we describe more realistic requirements for a progressive 
	quantum channel~\ref{subsec:progressChannel} which are sufficient for the ``forward-looking" analysis technique to work. 
	In \ref{subsec:keyLemma} we prove the key Lemma for progressive quantum channels and in \ref{subsec:upperBounds} we use the key Lemma to prove that the resampling algorithm (Alg.~\ref{alg:flawSelect}) works efficiently
	under Shearer's condition. In \ref{subsec:diffVersions} we draw the conclusions for various scenarios, while in \ref{subsec:runtimeCompare} we compare the results with classical algorithms.
	
	In Section~\ref{sec:efficient} we show how to efficiently implement a progressive quantum channel using only weak and strong measurements, and how to implement approximate versions of the required measurement operators. Finally we put all pieces together to prove the main theorem on the runtime of our quantum algorithm, 
	while in Section~\ref{sec:discussion} we consider a possible generalisation.
	\section{Definitions and notation}\label{sec:Definitions}
	
	In this work, for simplicity, we focus on the case of $2$-level systems (qubits), but all the results in this work apply equally well for $d$-level systems (qudits).
	Also, in order to maintain convenient formulation, pure states such as $\ket{\psi}$ and mixed states such as $\rho$ will not necessarily be normalised.
	
	\begin{definition} (Hilbert space of the qubits)
		Let \nom{$n$}{The number of qubits in the system} denote the number of qubits and let $N=2^n$, so that the \textbf{Hilbert space} of the quantum system is \nomenclature{$\mathcal{H}$}{The Hilbert space of the system}$\mathcal{H}=\mathbb{C}^N$. The qubits are labelled with elements from $[n]$. For $A\subseteq [n]$ let $\mathcal{H}_A$ denote the Hilbert space of the qubits in $A$ and $\Id_A$ denote the identity operator on this space.
	\end{definition}

	For the rest of the paper we are going to refer to the generalised ``quantum clauses", i.e., our projectors and their image, as flaws that we want to avoid. 
	
	\begin{definition} (Flaws as local projectors and the assigned probabilities)
		Let $F$ be a \textbf{set of }(labels of)\textbf{ flaws} and for each $f\in F$ let $\Pi_f$ be an orthogonal projector on $n$ qubits representing a $\{0,1\}$-valued \textbf{binary measurement operator} which indicates presence of flaw $f$ with measurement outcome~$1$. Later we also use labels $\{G,B\}$ for the measurement outcomes corresponding to ``Good" (i.e., $0$) and ``Bad" (i.e., $1$) outcomes respectively. 
		
		For all $f\in F$ there is a given subset of the qubits $b(f)\subseteq[n]$, such that the projector $\Pi_f$ acts trivially on $[n]\setminus b(f)$. 
		For $S\subseteq F$ we extend this notation by defining $b(S):=\bigcup_{f\in S} b(f)$. 
		
		Let $\Pi^{loc}_f$ denote $\Pi_f$ restricted to $b(f)$, so that we can write $\Pi_f=\Pi^{loc}_f\otimes \Id_{[n]\setminus b(f)}$.
		Let $p_f=\Tr(\Pi_f)/2^n=\Tr(\Pi_f^{loc})/2^{|b(f)|}$, and for $S\subseteq F$ let $p_S=\prod_{f\in S}p_f$. Note that $p_f$ is the \textbf{probability of measurement outcome} $1$ for $\Pi_f$ on a maximally mixed state.
	\end{definition}
	
	\begin{definition} (Dependency graph)
		Let us define the dependency graph $G=(F,E)$, where $E=\{\{f,f'\}:f\neq f' \text{ and }b(f)\cap b(f')\neq \varnothing \}$.
		For $f\in F$ let $\Gamma(f)=\{f'\in F: \{f,f'\}\in E\}$ denote the set of other flaws that overlap with $f$ and $\Gamma^+(f):=\Gamma(f)\cup \{f\}$. 	
		Similarly for $S\subseteq F$ let $\Gamma(S)=\bigcup_{f\in S} \Gamma(f)$ and $\Gamma^+(S)=\bigcup_{f\in S} \Gamma^+(f)$. Finally let $\Ind(F)=\{I\subseteq F: \Gamma(I)\cap I = \varnothing \}$ denote the set of independent sets of $G$.
	\end{definition}

	\begin{definition}\label{def:IndepPoly} (Independent set polynomial)
		Consider a vector of numbers $(x_f)_{f\in F}$. For every $I \in \Ind(F)$ we define the polynomial $q_I$ in the variables $(x_f)$  as follows:
		\begin{equation}
		q_I\left(x_f\right)=\sum_{S\in\Ind(F):\, I\subseteq S} (-1)^{|S|-|I|}\prod_{f\in S}x_f .
		\end{equation}
	\end{definition}

	\noindent Let $e=2.718\ldots$ denote the base of the natural logarithm in the following definition:
	\begin{definition}\label{def:Conditions} (Conditions)
		The vector of probabilities $(p_f)_{f\in F}$ is said to satisfy the 
		
		\noindent\textbf{Symmetric Lovász condition (SLC)}  if
		\begin{equation}\label{eq:SLC}
		\exists d\in\mathbb{N} \text{ s.t. } \forall f\in F\!:\, |\Gamma^+(f)|\leq d \text{ and } p_f\leq \frac{1}{d\cdot e} \tag{SLC}
		\end{equation}    
		
		\noindent\textbf{General Lovász condition (GLC)} if
		\begin{equation}\label{eq:ALC}
		\exists (x_f)_{f\in F} \text{ s.t. } \forall f\in F\!:\, x_f\in (0,1) \text{ and } \frac{p_f}{x_f}\leq \prod_{f'\in \Gamma(f)}(1-x_{f'}) \tag{GLC}
		\end{equation}
		
		\noindent\textbf{Cluster expansion condition (CEC)} if
		\begin{equation}\label{eq:CEC}
		\exists (y_f)_{f\in F} \text{ s.t. } \forall f\in F\!:\, y_f> 0 \text{ and } \frac{y_f}{p_f}\geq 
		\sum_{\underset{J \in \Ind(F)}{J\subseteq \Gamma^+(f)}}\prod_{f'\in J}y_{f'} \tag{CEC}
		\end{equation}
		
		\noindent\textbf{Shearer's condition (SHC)} if
		\begin{equation}\label{eq:SHC}
		q_\varnothing(p_f)>0 \text{ and }\, \forall I\in \Ind(F):\, q_I(p_f)\geq 0 \tag{SHC}
		\end{equation}
	\end{definition}
    
    Evaluating the Independent set polynomial is $\SharpP$-hard~\cite{Hoffmann10}. Nevertheless, due to the importance of the condition~\eqref{eq:SHC} to  repulsive lattice gas models in many-body physics~\cite{Scott2005}, it is well understood for many lattices~\cite{heilmann1972,Baxter1980,Gaunt1965,Gaunt1967,todo1999transfer}. The existential proof showing that the condition~\eqref{eq:SHC} implies frustration-freeness has been used to prove frustration-freeness of QSAT instances with various lattice topologies, and to derive new numerical lower-bounds on the SAT/UNSAT transition of quantum satisfiability on random Erdős–Rényi models~\cite{SattathLatice}.
	
	\begin{prop}\label{prop:ShearerImplied} 
		\eqref{eq:SLC} implies \eqref{eq:ALC}. Also
		each of \eqref{eq:ALC} and \eqref{eq:CEC} implies \eqref{eq:SHC}. 
	\end{prop}
	\begin{proof}
		Let $d$ be as in \eqref{eq:SLC} and let $x_f=1/d$ for all $f\in F$, then \eqref{eq:ALC} holds: 
		\begin{equation} \label{eq:SLCImplied}
		\forall f\in F\!:\, \frac{p_f}{x_f}=p_f\cdot d\leq \frac{1}{e} \leq \left(1-\frac{1}{d}\right)^{d-1}
		\leq\prod_{f'\in \Gamma(f)}(1-x_{f'}).
		\end{equation}	
		It is well known in the literature of the classical Lovász Local Lemma, that \eqref{eq:SHC} is a weaker condition than \eqref{eq:ALC} or \eqref{eq:CEC}.
		For direct proofs of these implications see, e.g., \cite[Corollary 5.37]{HarveyVondrak15Ar} and \cite[Corollary 5.42]{HarveyVondrak15Ar} respectively.
	\end{proof}
	
	
	\begin{definition} (Subspaces and projectors) 
		Whenever we refer to subspaces we always refer to subspaces of $\mathcal{H}_{[n]}$, in particular we define $\bigcap_{i\in\varnothing}V_i=\mathcal{H}_{[n]}$. For a subspace $V\subseteq\mathcal{H}_{[n]}$ let $\Pi_V$ denote the \textbf{orthogonal projector} to $V$.
	\end{definition}
	
	\begin{definition} (Semidefinite ordering) 
		Suppose $A,B$ are hermitian operators on $\mathcal{H}$. Then we write $A\preceq B$ if and only if $0\preceq B-A$ with the latter meaning by definition that $B-A$ is positive semidefinite.
	\end{definition}
	
	\begin{definition} (Sign function)
		For $x\in\mathbb{R}$ let us denote the \textbf{sign function} by
		\begin{equation*}
		\sgn{x}=
		\begin{cases}
		-1, & \text{for } x<0\\
		\phantom{-}0, & \text{for } x=0\\
		\phantom{-}1, & \text{for } x>0 .
		\end{cases}
		\end{equation*}
		For a diagonal matrix $\Sigma$ we define $\sgn{\Sigma}$ element-wise.
	\end{definition}
	
	\begin{definition} \label{def:qCStates}(Quantum-classical states)
		For the description of quantum-classical states consisting of an $N$ dimensional quantum system and a $k$ dimensional classical system we are going to use elements of $\mathbb{C}^{N\times N}\otimes \mathbb{R}^k$. We can interpret these as quantum states of restricted form via defining an embedding of $\mathbb{R}^k$ to $\mathbb{C}^{k\times k}$ using diagonal matrices.
	\end{definition}	
	
	\begin{definition} (Trace norm and distance)
		The trace norm of a matrix is the sum of its singular values: $\trnorm{M} = \Tr(\sqrt{M M^\dagger})$. The trace distance between compatible matrices $A,B$ is $\trnorm{A-B}$.
	\end{definition}
	For elements of $\mathbb{C}^{N\times N}\otimes \mathbb{R}^k$ we define the trace norm via the embedding of Definition~\ref{def:qCStates}. Moreover, if $\{e_i:i\in[k]\}$ is an orthonormal basis of $\mathbb{R}^k$, and $M=\sum_{i\in[k]}M_i\otimes e_i$, then $\trnorm{M}=\sum_{i\in[k]}\trnorm{M_i}$.
	
	\begin{definition} (Approximate quantum channels)
		Let $\mathcal{Q}, \tilde{\mathcal{Q}}$ be quantum channels between the spaces $:\mathbb{C}^{N\times N}\otimes \mathbb{R}^k\rightarrow \mathbb{C}^{M\times M}\otimes \mathbb{R}^\ell$. We say that $\tilde{\mathcal{Q}}$ $\delta$-approximates $\mathcal{Q}$, if 
		\begin{equation}
		\sup\left\{\trnorm{\mathcal{Q}(A)-\tilde{\mathcal{Q}}(A)}
		: A\in \mathbb{C}^{N\times N}\otimes \mathbb{R}^k, A=A^\dagger, \trnorm{A}\leq 1
		\right\}\leq \delta.
		\end{equation}
        \label{def:approximate_channel}
	\end{definition}	
   
	\section{The algorithm and the key lemma for the analysis}\label{sec:forward}
	
	The following algorithm is inspired by the classical MaximalSetResample algorithm from Harvey and Vondrák~\cite{HarveyVondrak15}.
	To adapt the algorithm to the quantum setting we introduce a quantum channel $\mathcal{Q}_f^C$, which performs some quantum operation on the $n$-qubit quantum register determined by the classical input $(C,f)$, where $C$ is the set of already ``checked" flaws, and $f$ is the next flaw to address.
	In the case of commuting projectors $\mathcal{Q}_f^C$ will be simply the application of a projective measurement $(\Pi_f,\Id-\Pi_f)$ where the measurement outcomes are labelled with $(B,G)$ standing for (``Bad",``Good") respectively.
	\begin{algorithm}[H]
		\caption{Maximal independent set resampling algorithm
		}\label{alg:flawSelect}
		\begin{algorithmic}[1]
			\STATE set $C\leftarrow\varnothing$ \kern10mm($\star\,\, C\subseteq F$ is going to store the checked flaws which were found ``good" $\,\,\star$)
			\STATE start with the maximally mixed state $\rho_0=\Id/2^n$
			\STATE {\bf while} $C\neq F$ {\bf do} \label{line:mainLoopStart}
			\STATE ~~~~ set $I\leftarrow\varnothing$ \kern22mm($\star\,\, I\subseteq F$ is going to store the flaws resampled in this round $\,\,\star$)
			\STATE ~~~~ {\bf while} $C\cup\Gamma^+(I)\neq F$ {\bf do} \label{line:indepCondition}	
			\STATE ~~~~~~~~ pick the minimal $f\in F\setminus (C\cup\Gamma^+(I))$ according to some fixed ordering $(\leq,F)$ \label{line:deterministicSelection}	
			\STATE ~~~~~~~~ apply quantum channel $\mathcal{Q}_f^C$ \label{line:applyQ}
			\kern1mm($\star\,\, $ Use $\mathcal{Q}_f^C=$``measure $\Pi_f$" in the commuting case $\,\,\star$) $\phantom{-}$
			\kern1mm($\star\,\, $ For non-commuting operators $\mathcal{Q}_f^C$ is Algorithm~\ref{alg:measurement} with a small error parameter $\theta$ $\,\,\star$)\kern0mm
			\STATE ~~~~~~~~ {\bf if} outcome label is $E$ {\bf then terminate with} ``ERROR" \label{line:outcomeIf}	
			\STATE ~~~~~~~~ {\bf else if} outcome label is $G$ {\bf then} set
			$C \leftarrow C \cup \{f\}$				
			\STATE ~~~~~~~~ {\bf else if} outcome label is $B$ {\bf then}
			\STATE ~~~~~~~~~~~~ set $I\leftarrow I\cup \{f\}$
			\STATE ~~~~~~~~~~~~ set $C\leftarrow C\setminus \Gamma(f)$	
			\STATE ~~~~~~~~~~~~ resample $f$ \kern12mm($\star\,\, $ Replace qubits $b(f)$ with maximally mixed qubits $\,\,\star$) \label{line:resample}	
			\STATE ~~~~~~~~ {\bf end if	} \label{line:outcomeIfEnd}	
			\STATE ~~~~ {\bf end while}		
			\STATE {\bf end while} \label{line:mainLoopEnd}
			\STATE {\bf terminate with} ``SUCCESS" 		
		\end{algorithmic}
	\end{algorithm}
	
	\begin{definition} 
		For $f\in F$ the \textbf{resampling operation} on $\rho$ in line~\ref{line:resample} can be formally described as $R_f(\rho)=\Tr_{b(f)}[\rho]\otimes\Id_{b(f)}\cdot 2^{-|b(f)|}$.
	\end{definition}
	
	\begin{definition} 
		Lines~\ref{line:mainLoopStart}-\ref{line:mainLoopEnd} will be called a \textbf{round}. Due to the selection rule of $f$ in line~\ref{line:deterministicSelection}, $I$ will always be an independent set. We will denote the independent set at the end of the $i$-th round by $I_i$. 
	\end{definition}	
	
	\subsection{Required properties of the quantum channel \texorpdfstring{$\mathcal{Q}$}{Q}}\label{subsec:progressMeas}
	In this subsection we define some properties of $\mathcal{Q}$ under which we can analyse Algorithm~\ref{alg:flawSelect} nicely.
	Later we show that Algorithm~\ref{alg:measurement} satisfies these requirements.
	
	In order to prove that Algorithm~\ref{alg:flawSelect} converges to a good quantum state, we would like $\mathcal{Q}$ not to disturb the ``good" part of the quantum state, but efficiently project out its ``bad" part. 
	When $\mathcal{Q}$ fails to project out the ``bad" part, it should transform the ``bad" part of the state to the image of the failed  projector, for reasons that will be explained later.
	
	We define progress in terms of ``good" subspaces, because we want to ensure that the flaws that are marked as checked (denoted by $C\subseteq F$) are indeed satisfied.
	We hope for an algorithm that works for low-energy subspaces as well, not just for zero-energy ones, and the high-level analysis works in this case as well -- see Appendix~\ref{apx:fpras}. This is why we define the somewhat abstract concept of a subspace progress measure. However, we will be mostly concerned about zero-energy subspaces, and thus use the exact progress measure as defined below.
	
	\begin{definition} 
		We call $\{V^C:C\subseteq F\}$ a \textbf{subspace progress measure} if for all $C\subseteq F:$  $V^C$ is a subspace of $\mathcal{H}_{[n]}$, and $V^C$ is only dependent on qubits $b(C)$, i.e., there is 
		$\tilde{V}^C\subseteq \mathcal{H}_{b(C)}$ such that $V^C=\tilde{V}^C\otimes \mathcal{H}_{[n]\setminus b(C)}$.
	\end{definition}
	
	We will mostly be concerned with the following natural subspace progress measure, aiming at frustration-free states.
	In the following definition $V^F$ is the kernel of $H$, the subspace to which we would like to gradually converge.
	\begin{definition} \label{def:exactMeasure}
		We call $\{V^C=\bigcap_{f'\in C}\ker(\Pi_{f'}):C\subseteq F\}$ the \textbf{exact progress measure}.
	\end{definition}
	
	When we are only concerned with termination of Algorithm~\ref{alg:flawSelect} we will consider the following trivial subspace progress measure:
	\begin{definition} \label{def:trivialMeasure}
		We call $\{V^C=\mathcal{H}_{[n]}:C\subseteq F\}$ the \textbf{trivial progress measure}.
	\end{definition}
	
	\subsubsection{The exact quantum channel -- ideal non-commuting generalisation}\label{subsec:idealChannel}
	
	Now we introduce our generalisation of the measurement procedure for the non-commuting setting.
	We argue that this is probably the most faithful generalisation of the commuting algorithm for the non-commuting case.
	The proposed quantum operation applies a measurement conditionally followed by a unitary operation. The combined procedure both respects the loop-invariant of the exact progress measure, and handles new flaws in a way which seems essential for the resampling algorithm. (In the following definition we use notation $V^C=\bigcap_{f'\in C}\ker(\Pi_{f'})$ corresponding to the exact progress measure.)
	
	\begin{definition} \label{def:exactChannel}
		We define the \textbf{exact quantum channel}, denoted here by $\mathcal{Q}$, in  the following way: conditional on receiving classical information $C\subseteq F$ and $f\in F$, the quantum channel $\mathcal{Q}^C_f:\mathbb{C}^{N\times N}\rightarrow\mathbb{C}^{N\times N}\otimes \mathbb{R}^2$ performs the projective measurement $\left(\Pi_{V^{C\cup \{f\}}}, \Id-\Pi_{V^{C\cup \{f\}}}\right)$. If the outcome is $\Pi_{V^{C\cup \{f\}}}$ it labels its output with $G$ standing for ``Good". If the outcome is $\Id-\Pi_{V^{C\cup \{f\}}}$ it labels its output with $B$ standing for ``Bad", then it applies the unitary operation $\mathrm{Rot}=WU^\dagger$, where $W\Sigma U^\dagger$ is a singular value decomposition of $\Pi_f \Pi_{V^C}$.
		For the output state corresponding to pure input state $\ket{\psi}$ we use notation $\mathcal{Q}^C_f\left(\ket{\psi}\right)=\ket{\psi_G}\otimes G+\ket{\psi_B}\otimes B$, where $\ket{\psi_G}=\Pi_{V^{C\cup \{f\}}}\ket{\psi}$ and $\ket{\psi_B}=W U^\dagger \left(\Id-\Pi_{V^{C\cup \{f\}}}\right)\ket{\psi}$.
	\end{definition}
	\begin{remark}
		In the above definition we have some ambiguity about the map $WU^\dagger$, since the singular value decomposition is not unique.
		However one can show that the map $W\sgn{\Sigma}U^\dagger$ is well defined (see Appendix~\ref{apx:uniqueMap}), and this is enough in our case,
		since we are always acting on input states that lie in $\Pi_{V^C}$.
		If $\ket{\psi}=\Pi_{V^C}\ket{\psi}$, then 
		$\left(\Id-\Pi_{V^{C\cup \{f\}}}\right)\ket{\psi}
		=\left(\Id-\Pi_{V^{C\cup \{f\}}}\right)\Pi_{V^C}\ket{\psi}
		=\left(\Pi_{V^C}-\Pi_{V^{C\cup \{f\}}}\right)\ket{\psi}$.
		But then the action of $WU^\dagger$ only depend on $W\sgn{\Sigma}U^\dagger$, as we can show using \eqref{eq:badNonOrth} from the next Proposition:  $WU^\dagger\left(\Pi_{V^C}-\Pi_{V^{C\cup \{f\}}}\right)=WU^\dagger\left(U\sgn{\Sigma} U^\dagger\right)=W\sgn{\Sigma}U^\dagger U\sgn{\Sigma} U^\dagger=W\sgn{\Sigma}U^\dagger\left(\Pi_{V^C}-\Pi_{V^{C\cup \{f\}}}\right)$.
	\end{remark}
	
	To justify our generalisation and explain better its main purpose we show that this quantum channel	preserves two important properties of the commuting quantum case. Before we prove the corresponding Lemma~\ref{lemma:exactChannel}, we need some identities of the relevant subspaces. In the following we use concise rank arguments, but the reader may get more insight on the structure of the examined subspaces by looking at Jordan's Theorem in Appendix~\ref{apx:Jordan}.
	
	\begin{prop} \label{prop:subspaceIdentities}
		Let us fix some $f\in F$ and $C\subseteq F$. Let us use notation $V=\bigcap_{f'\in C}\ker(\Pi_{f'})$, $V_f=V\cap\ker(\Pi_f)$ and $V_{\overline{f}}=\bigcap_{f'\in C\setminus \Gamma^+(f)}\ker(\Pi_{f'})$. Suppose $W\Sigma U^\dagger$ is a singular value decomposition of $\Pi_f \Pi_V$ (i.e., $\Pi_f \Pi_V=W\Sigma U^\dagger$ with $W^\dagger=W^{-1}$, $U^\dagger=U^{-1}$ and $\Sigma$ diagonal), then the following identities hold:
		\begin{align}
		\Pi_{\im(\Pi_f\Pi_V)}=& \ W\sgn{\Sigma} W^\dagger \label{eq:imAdj}\\
		\Pi_V-\Pi_{V_f}=& U\sgn{\Sigma} U^\dagger \label{eq:badNonOrth}\\
		\Pi_{\im(\Pi_f\Pi_V)}\preceq&  \Pi_f \Pi_{V_{\overline{f}}} \label{eq:imUpBnd}
		\end{align}
	\end{prop}
	\begin{proof}
		\eqref{eq:imAdj}: $\left(W\sgn{\Sigma} W^\dagger\right)\Pi_f \Pi_V
		=\left(W\sgn{\Sigma} W^\dagger\right)W\Sigma U^\dagger
		=W\Sigma U^\dagger=\Pi_f \Pi_V$, since $W\sgn{\Sigma} W^\dagger$ is an orthogonal projector it implies $\Pi_{\im(\Pi_f\Pi_V)}\preceq \ W\sgn{\Sigma} W^\dagger$.
		But also $\rank{\Pi_{\im(\Pi_f\Pi_V)}}=\rank{\Pi_f\Pi_V}=\rank{W\sgn{\Sigma} W^\dagger}$, 
		thus $\Pi_{\im(\Pi_f\Pi_V)}= \ W\sgn{\Sigma} W^\dagger$.
		
		\eqref{eq:badNonOrth}: Similarly to \eqref{eq:imAdj} 
		$\Pi_{\im(\Pi_V\Pi_f)}= \ U\sgn{\Sigma} U^\dagger$, so it is enough to show that
		$\Pi_V-\Pi_{V_f}=\Pi_{\im(\Pi_V\Pi_f)}$. 
		Again $\Pi_{\im(\Pi_V\Pi_f)}\preceq\Pi_V-\Pi_{V_f}$, since
		$\left(\Pi_V-\Pi_{V_f}\right)\Pi_V\Pi_f=\Pi_V\Pi_f-\Pi_{V_f}\Pi_f=\Pi_V\Pi_f$.
		
		But 
        \begin{align*}
        \rank{\Pi_{\im(\Pi_V\Pi_f)}}
	&	=\rank{\Pi_V\Pi_f}\\
	&	=\rank{\Pi_f\Pi_V}\\
	&	=\rank{\Pi_V} - \dim(\ker(\Pi_f) \cap \im(\Pi_V)\\
	&	=\dim(V)-\dim(V_f)\\
	&	=\rank{\Pi_V-\Pi_{V_f}}.
        \end{align*}
		Here, the second equality is justified by $\rank{A}=\rank{A^\dagger}$, and the third equality by $\rank{AB}=\rank{B}-\dim(\ker(A) \cap \im(B) )$ (see, e.g. \cite[p. 210]{meyer00matrix}). 
		
		So $\Pi_V-\Pi_{V_f}=\Pi_{\im(\Pi_V\Pi_f)}$ and thus $\Pi_V-\Pi_{V_f}= U\sgn{\Sigma} U^\dagger$.
		
		\eqref{eq:imUpBnd}: The proof follows form the following line of (in)equalities which we justify below:
		$$\Pi_{\im(\Pi_f\Pi_V)}
		=\Pi_f \Pi_{\im(\Pi_f\Pi_V)} \Pi_f
		\preceq\Pi_f \Pi_{V_{\overline{f}}} \Pi_f
		= \Pi^2_f \Pi_{V_{\overline{f}}}
		= \Pi_f \Pi_{V_{\overline{f}}}.$$
		First observe that $\Pi_f\left(\Pi_f\Pi_V\right)=\Pi_f\Pi_V$ so $\Pi_f\Pi_{\im(\Pi_f\Pi_V)}=\Pi_{\im(\Pi_f\Pi_V)}$, implying the first equality. 
		The penultimate equality is due to $\Pi_f\Pi_{V_{\overline{f}}}=\Pi_{V_{\overline{f}}}\Pi_f$, which follows from the fact that these operators act on disjoint qubits. 
		Finally note that $\Pi_{V_{\overline{f}}}\Pi_V=\Pi_V$, since $V\subseteq V_{\overline{f}}\,$. Therefore, 		$\Pi_{V_{\overline{f}}}\left(\Pi_f\Pi_V\right)
		=\Pi_f\Pi_{V_{\overline{f}}}\Pi_V
		=\Pi_f\Pi_V$ so $\Pi_{\im(\Pi_f\Pi_V)}\preceq\Pi_{V_{\overline{f}}}$, which justifies the inequality.
   	\end{proof}

	\noindent Using the above proposition we can easily show in the following lemma that the properties of our interest hold. However one might be puzzled why is it important to transform states to the ``Bad" image of $\Pi_f$. The reason is that it ensures that the resampling operation uniformly mixes quantum states, for more details see the proof of Lemma~\ref{lemma:NonComm}; Appendix~\ref{apx:WhyRotate} presents a simple example showing how things can go wrong when this transformation step is skipped.

	\begin{lemma} \label{lemma:exactChannel}
		Suppose $\ket{\psi}$ respects the exact progress measure with respect to checked flaws $C\subseteq F$, i.e., $\ket{\psi}=\Pi_{V^C}\ket{\psi}$. 
		If we apply the exact quantum channel $\mathcal{Q}^C_f$ on $\ket{\psi}$, then 
		\begin{enumerate}[label=(\textbf{\roman*})]
			\item[$\ket{\psi_G}$\kern-0.5mm] lies in $V^{C\cup\{f\}}$, i.e., $\ket{\psi_G}=\Pi_{V^{C\cup\{f\}}}\ket{\psi_G}$, and \label{it:ExactGood}
			\item[$\ket{\psi_B}$\kern-0.5mm] lies in $\im(\Pi_f)\cap V^{C\setminus\Gamma^+(f)}$, i.e., $\ket{\psi_B}=\Pi^{loc}_f\otimes \Pi'_{V^{C\setminus\Gamma^+(f)}}\ket{\psi_B}$\\
			(where $\Pi_{V^{C\setminus\Gamma^+(f)}}=\Id_{b(f)}\otimes \Pi'_{V^{C\setminus\Gamma^+(f)}} $).
			\label{it:ExactBad}
		\end{enumerate}
	\end{lemma}
	\begin{proof}
		By Definition~\ref{def:exactChannel}  $\ket{\psi_G}=\Pi_{V^{C\cup\{f\}}}\ket{\psi_G}$, so the first property is trivial.
		
		By Definition~\ref{def:exactChannel}  $\ket{\psi_B}=W U^\dagger \left(\Id-\Pi_{V^{C\cup \{f\}}}\right)\ket{\psi}$. Note that by $\ket{\psi}=\Pi_{V^C}\ket{\psi}$ we have 
		\begin{equation}\label{eq:subProjRespect}
			\left(\Id-\Pi_{V^{C\cup\{f\}}}\right)\ket{\psi}
			=\left(\Id-\Pi_{V^{C\cup\{f\}}}\right)\Pi_{V^{C}}\ket{\psi}
			=\left(\Pi_{V^{C}}-\Pi_{V^{C\cup\{f\}}}\right)\ket{\psi}.
		\end{equation}
		Using \eqref{eq:badNonOrth} we can see
		$WU^\dagger \left(\Pi_{V^{C}}-\Pi_{V^{C\cup\{f\}}}\right)
		= WU^\dagger U\sgn{\Sigma} U^\dagger
		= W \sgn{\Sigma} U^\dagger$. 
		Considering $\left(\sgn{\Sigma}\right)^2=\sgn{\Sigma}$ and $U^\dagger U=\Id$ we get
		$W\sgn{\Sigma}U^\dagger=W\sgn{\Sigma}W^\dagger WU^\dagger U\sgn{\Sigma}U^\dagger$
		and by \eqref{eq:imAdj}-\eqref{eq:badNonOrth} we get 
		$W\sgn{\Sigma}W^\dagger WU^\dagger U\sgn{\Sigma}U^\dagger = \Pi_{\im(\Pi_f\Pi_{V^C})}WU^\dagger\left(\Pi_{V^{C}}-\Pi_{V^{C\cup\{f\}}}\right)$.
		Therefore, we proved $WU^\dagger \left(\Pi_{V^{C}}-\Pi_{V^{C\cup\{f\}}}\right) 
		= \Pi_{\im(\Pi_f\Pi_{V^C})}WU^\dagger\left(\Pi_{V^{C}}-\Pi_{V^{C\cup\{f\}}}\right)$. 
		By~\eqref{eq:imUpBnd} we have $\Pi_{\im(\Pi_f\Pi_{V^C})}\preceq  \Pi_f \Pi_{V^{C\setminus\Gamma^+(f)}}=\Pi^{loc}_f\otimes \Pi'_{V^{C\setminus\Gamma^+(f)}}$ which implies
		$W U^\dagger \left(\Pi_{V^C}-\Pi_{V^{C\cup \{f\}}}\right)
		= \left(\Pi^{loc}_f\otimes \Pi'_{V^{C\setminus\Gamma^+(f)}}\right)
		WU^\dagger\left(\Pi_{V^C}-\Pi_{V^{C\cup \{f\}}}\right)$ proving	$\ket{\psi_B}=\Pi^{loc}_f\otimes \Pi'_{V^{C\setminus\Gamma^+(f)}}\ket{\psi_B}$ via \eqref{eq:subProjRespect}.
	\end{proof}
	
	\noindent For completeness we show that Definition~\ref{def:exactChannel} is indeed a generalisation of the commuting case.
	\begin{prop} \label{prop:indeedGeneralisation}
		Suppose that all local projectors commute, and that the input state $\ket{\psi}$ is such that $\ket{\psi}=\Pi_{V^C}\ket{\psi}$, then the output of the exact quantum channel $\mathcal{Q}^C_f$ coincides with the output of the projective measurement $(\Id-\Pi_f,\Pi_f)$
		, i.e., $\ket{\psi_G}=\left(\Id-\Pi_f\right)\ket{\psi}$ and  $\ket{\psi_B}=\Pi_f\ket{\psi}$.
	\end{prop}
	\begin{proof}
		Since all local projectors commute we have $\Pi_{V^C}=\prod_{f'\in C}\left(\Id-\Pi_{f'}\right)$. 
		By Definition~\ref{def:exactChannel} $\ket{\psi_G}=\Pi_{V^{C\cup \{f\}}}\ket{\psi}$ and due to commutation we have $\Pi_{V^{C\cup \{f\}}}=(\Id-\Pi_f)\Pi_{V^C}$, so $\ket{\psi_G}=(\Id-\Pi_f)\Pi_{V^C}\ket{\psi}=(\Id-\Pi_f)\ket{\psi}$.
		
		By Definition~\ref{def:exactChannel} 
		$\ket{\psi_B}=W U^\dagger \left(\Id-\Pi_{V^{C\cup \{f\}}}\right)\ket{\psi}$, furthermore similarly to the proof of Lemma~\ref{lemma:exactChannel}
		$\left(\Id-\Pi_{V^{C\cup \{f\}}}\right)\ket{\psi}
		=\left(\Id-\Pi_{V^{C\cup \{f\}}}\right)\Pi_{V^{C}}\ket{\psi}
		=\left(\Pi_{V^{C}}-\Pi_{V^{C\cup \{f\}}}\right)\ket{\psi}$ by our assumption on $\ket{\psi}$. Using \eqref{eq:badNonOrth} we get that $\ket{\psi_B}=W U^\dagger U\sgn{\Sigma} U^\dagger \ket{\psi}=W\sgn{\Sigma} U^\dagger \ket{\psi}$.
		By commutation we have that $\Pi_f\Pi_{V^C}=\Pi_{V^C}\Pi_f$ is an orthogonal projector and thus $\Sigma=\sgn{\Sigma}$. Therefore, $W\sgn{\Sigma} U^\dagger=W\Sigma U^\dagger=\Pi_f\Pi_{V^C}$ and thus $\ket{\psi_B}=\Pi_f\Pi_{V^C}\ket{\psi}=\Pi_f\ket{\psi}$.
	\end{proof}		
	
	\begin{remark} \label{rem:couldCompress}
		The properties proven in Lemma~\ref{lemma:exactChannel} enable one to use the exact quantum channel of Definition~\ref{def:exactChannel} in the non-commuting setting together with the compression argument of \cite{SchwarzInfo,SattathSymm} to show that, under the condition \eqref{eq:SLC}, the algorithms in \cite{SchwarzInfo,SattathSymm}  find a ground state quickly.
	\end{remark}

	For now we do not continue in the direction of Remark~\ref{rem:couldCompress} for two reasons. The first reason is, that we do not know how to prove efficiency up Shearer's bound \eqref{eq:SHC} using a compression argument.
	The second reason is that we cannot implement the exact quantum channel efficiently. 
	In Section~\ref{sec:efficient} we show how to efficiently implement a closely related quantum channel.
	That quantum channel fits the proof techniques of the ``forward-looking" analysis, 
	however, it has the drawback that even if the input is a pure quantum sate its outputs can in general only be described by a probabilistic mixture of pure states. Although this feature can also be handled using the entropy compression argument of \cite{SattathSymm}, it requires some additional analysis of the procedure. The main advantage of using the compression argument is, that it improves the runtime bound in the \eqref{eq:SLC} case, therefore we plan to work out the details in the next arXiv version of this paper.

	
	\subsubsection{Progressive quantum channels -- an efficient non-commuting version}\label{subsec:progressChannel}
	
	Because of the probabilistic mixtures appearing we can no longer work with the convenient pure state formalism,
	so from now on, we will use density operators instead. The following definition formulates the requirements of 
	the ``forward-looking" analysis technique in terms of density operators. The most important criterion is that if the input is a (mixed) quantum state~$\rho$, which is supported on some subspace respecting the loop-invariant and is upper bounded by a uniform mixture, then the output state should also lie in a nice subspace and should also be upper bounded by a corresponding uniform mixture. Since we can only implement approximate versions of our channels, there is some other criterion concerning error probabilities.

	\begin{definition} 
		\label{def:progressive_channel}
		We say that $\mathcal{Q}$ is a \textbf{progressive quantum channel} with respect
		to the subspace progress measure $\{V^C:C\subseteq F\}$ with error parameter $\theta\in[0,1]$ if the following holds:
		Conditional on receiving classical information $C\subseteq F$ and $f\in F$, the quantum channel $\mathcal{Q}^C_f$ performs the quantum operation $\mathcal{Q}^C_f:\mathbb{C}^{N\times N}\rightarrow\mathbb{C}^{N\times N}\otimes \mathbb{R}^3$, satisfying the following properties: 
		
		(In the following definition think about $\Pi_{V}$ as an unnormalised quantum state representing uniform distribution on $V$.)
		
		\begin{enumerate}[label=(\textbf{\roman*})]
			\item The quantum channel labels its output with the classical labels $(G,B,E)$ corresponding to $(\text{``Good",``Bad",``Error"})$ outcomes, so that for input $\rho$ the output state is written as:\\
			$\mathcal{Q}^C_f(\rho)=\mathcal{Q}^C_{f,G}(\rho)\otimes G+
			\mathcal{Q}^C_{f,B}(\rho)\otimes B+\mathcal{Q}^C_{f,E}(\rho)\otimes E$. \label{it:Aoutcomes}
			\item For a uniformly mixed input state on $\Pi_{V^C}$, the output state labelled as ``Good" is upper bounded by a uniform mixture on the better subspace $\Pi_{V^{C\cup\{f\}}}$: \\
			$\mathcal{Q}^C_{f,G}(\Pi_{V^C})\preceq\Pi_{V^{C\cup\{f\}}}$. \label{it:Agood}
			\item For a uniformly mixed input state on $\Pi_{V^C}$, the output state labelled as ``Bad" is upper bounded by a uniform mixture on a subspace of tensor product form:
			 \\
			$\mathcal{Q}^C_{f,B}(\Pi_{V^C})\preceq \Pi^{loc}_f\otimes \Pi'_{V^{C\setminus\Gamma^+(f)}}$, 
			where $\Pi_{V^{C\setminus\Gamma^+(f)}}=\Id_{b(f)}\otimes \Pi'_{V^{C\setminus\Gamma^+(f)}} $.
			\label{it:Abad}
			\item For $\rho$ supported on $V^C$ (i.e., $\rho=\Pi_{V^C}\rho\Pi_{V^C}$), this has small ``Error": \\
			$\,\tr{\mathcal{Q}^C_{f,E}(\rho)}\leq 2\theta\cdot \tr{\mathcal{Q}^C_{f,B}(\rho)}$. \label{it:Aerror}
		\end{enumerate}
	\end{definition}
	At this point the requirement $\,\tr{\mathcal{Q}^C_{f,E}(\rho)}\leq 2\theta\cdot \tr{\mathcal{Q}^C_{f,B}(\rho)}$ may feel weird,	in fact it would be enough to require $\,\tr{\mathcal{Q}^C_{f,E}(\rho)}\leq 2\theta\cdot \tr{\mathcal{Q}^C_f(\rho)}$, 
	but our channels happen to satisfy the stronger condition, which results in slightly better runtime bounds. The scaling of the error parameter $\theta$ by $2$ is somewhat arbitrary, but it fits nicely with the analysis of Algorithm~\ref{alg:measurement}.
	
	\begin{prop} \label{prop:exactProgressive}
		The exact quantum channel is a progressive quantum channel with respect to the exact progress measure with error parameter $\theta=0$.
	\end{prop}
	
	This proposition is a direct corollary of Lemma~\ref{lemma:exactChannel}, and it shows that progressive quantum channels are indeed generalisations of the ideal exact channel we described before. 
		
		

	\subsection{The key lemma for the high-level analysis}\label{subsec:keyLemma}
	In this subsection we prove the key lemma. To do so we need to define some more concepts.
	
	\begin{definition}
		Let $L\in\{G,B,E\}^T$ be the \textbf{measurement log} of Algorithm~\ref{alg:flawSelect} after $T$ applications of the quantum channel $\mathcal{Q}$, where $(G,B,E)$ stands for $(\text{``Good",``Bad",``Error"})$ respectively.
		Let $\mathcal{L}_T\subseteq\{G,B,E\}^T$ denote the set of length-$T$ \textbf{valid logs}, containing all the measurement logs that Algorithm~\ref{alg:flawSelect} may see, if we allow all $3$ possible outcomes to happen with non-zero probability, and let $\mathcal{L}=\cup_{T=0}^{\infty}\mathcal{L}_T$. 
		For $X\in \{B,E\}$ let $\mathcal{L}^X\!=\!\{L\in\mathcal{L},\text{ such that } L \text{ ends with ``}X\text{"}\}$.
		
		Finally let $\mathcal{L}^{B(k)}=\{L\in\mathcal{L}^B: L \text{ contains } $k$ \text{ occurences of } B\}$.
	\end{definition}
	Note that there are logs that the algorithm can never encounter. 
	For example, a log may contain only one $E$, and it should be the last letter, 
	since Algorithm~\ref{alg:flawSelect} immediately terminates after measurement outcome $E$.
	Also we cannot see $|F|+1$ consecutive ``$G$" because seeing $|F|$ ``$G$" in a row would necessarily set $C=F$, and thus would terminate the algorithm.
	
	Since Algorithm~\ref{alg:flawSelect} is deterministic apart from the labels obtained via the application of $\mathcal{Q}$, we can completely reconstruct the inner variables $C$ and $I$ of the algorithm for any given log. For $L\in \mathcal{L}$ let $C_L$ and $I_L$ denote the inner variables $C$ and $I$ of Algorithm~\ref{alg:flawSelect} after it has seen and processed the measurement results described by $L$, i.e., including the changes made to $C$ and $I$ in lines~\ref{line:outcomeIf}-\ref{line:outcomeIfEnd} of Algorithm~\ref{alg:flawSelect}.
	Also let $\rho_L$ denote the unnormalised quantum state after having seen and processed all measurement results in $L$, i.e., including the resampling step in line~\ref{line:resample} if the last result was ``$B$".
	
	Let $p_L=\prod_{f\in F}p_f^{L_f}$, where $L_f$ denotes the number of resamplings of $f$ given that we see $L\in \mathcal{L}$.
	Also for $X \in \{G,B,E\}$ let $(L,X) \in \{G,B,E\}^{T+1}$ be the log obtained by adding $X$ to the end of log $L$. If the algorithm did not terminate after $L$, i.e. $C_L\neq F$, $E\notin L$ then let $f_L$ denote the next flaw Algorithm~\ref{alg:flawSelect} will address.
	
	\begin{lemma} \label{lemma:NonComm} (Key lemma)
		If $\mathcal{Q}$ is a progressive quantum channel 
		with respect to the subspace progress measure $\{V^C:C\subseteq F\}$ with error parameter $\theta \in [0,1]$,
		then for every $L\in \mathcal{L}\setminus \mathcal{L}^E$ we have
		\begin{equation}\label{eq:rhoSub}
		\rho_L \preceq p_L\cdot\frac{\Pi_{V^{C_L}}}{N}.
		\end{equation}
		Moreover if $(L,E)\in \mathcal{L}$, then
		\begin{equation}\label{eq:lowErr}
		\Tr[\rho_{(L,E)}]\leq 2\theta \cdot\Tr[\rho_{(L,B)}].
		\end{equation}
	\end{lemma}
	\begin{proof}
		We prove \eqref{eq:rhoSub} and \eqref{eq:lowErr} for all $L\in \mathcal{L}_T$ by induction on $T$. For $T=0$ we have $\mathcal{L}_0=\{\varnothing\}$. $\rho_\varnothing=\rho_0=\Id/N$, $\Pi_{V^\varnothing}=\Id$ and $p_\varnothing=1$ so the relation holds with equality. Now suppose that \eqref{eq:rhoSub} holds for all $L\in\mathcal{L}_T\setminus \mathcal{L}^E$. For the induction step it is enough to show that \eqref{eq:rhoSub} also holds for $(L,G)$ and 
		$(L,B)$, whenever $(L,G)$ and $(L,B)$ are in $\mathcal{L}_{T+1}$. Let us denote $C=C_L$, $C_G=C_L\cup\{f_L\}$, $C_B=C_L\setminus\Gamma(f_L)$ and $f=f_L$. Observe $C_{(L,G)}=C_G$ and $C_{(L,B)}=C_B$.	
		First we show the inductive step for $(L,G)$:
		\begin{align*}
		\rho_{(L,G)}=&\mathcal{Q}^C_{f,G}(\rho_L) & & (\text{by definition})\\
		\preceq& \mathcal{Q}^C_{f,G}\left(p_L\cdot\frac{\Pi_{V^{C}}}{N} \right)  & & (\text{by the inductive hypothesis})\\
		\preceq& p_L\cdot\frac{\Pi_{V^{C_G}}}{N}   & & (\text{by property }\ref{it:Agood})\\
		=& p_{(L,G)}\cdot\frac{\Pi_{V^{C_{(L,G)}}}}{N} 
		& & (p_{(L,G)}=p_L;C_{(L,G)}=C_L )
		\end{align*}
		Now we show the inductive step for $(L,B)$:
		\begin{align*}
		\rho_{(L,B)}=&R_{f}\left(\mathcal{Q}^C_{f,B}(\rho_L)\right) & & (\text{by definition})\\
		\preceq& R_{f}\left(\mathcal{Q}^C_{f,B}\left(p_L\cdot\frac{\Pi_{V^{C}}}{N} \right)\right)  & & (\text{by the inductive hypothesis})\\
		\preceq& \frac{p_L}{2^n}\cdot
		R_{f}\left(\Pi^{loc}_{f}\otimes \Pi'_{V^{C_B}}\right) & & (\text{by property }\ref{it:Abad})\\
		=& \frac{p_L}{N}\cdot\frac{\Tr[\Pi^{loc}_{f}]}{2^{|b(f)|}}\cdot 
		\Id_{b(f)}\otimes \Pi'_{V^{C_B}} & & (\text{because }R_f\text{ acts locally} )\\
		=& \frac{p_L}{N}\cdot p_f\cdot 
		\Id_{b(f)}\otimes \Pi'_{V^{C_B}} & & (\text{by the definition of }p_f )\\
		=& \frac{p_L}{N}\cdot p_f\cdot 	\Pi_{V^{C_B}} & & (\text{by property }\ref{it:Abad})\\
		=& \frac{p_{(L,B)}}{N}\cdot \Pi_{V^{C_{(L,B)}}} & & (p_{(L,B)}=p_L\cdot p_f)	
		\end{align*}	
		For the proof of \eqref{eq:lowErr}, first note that $(L,E)\in \mathcal{L}$ implies that Algorithm~\ref{alg:flawSelect} does not terminate after seeing $L$ and thus also $(L,B)\in \mathcal{L}$.
		Finally due to property \ref{it:Aerror} of $\mathcal{Q}$, we have
		\begin{equation*}
		\tr{\rho_{(L,E)}}=\tr{\mathcal{Q}^C_{f,E}(\rho_L)}
		\leq 2\theta\cdot \tr{\mathcal{Q}^C_{f,B}(\rho_L)}
		=2\theta \cdot \tr{\rho_{(L,B)}}.^{}
		\end{equation*}
	\end{proof}	
	
	\subsection{Upper bounds by stable set sequences}\label{subsec:upperBounds}
	In Appendix~\ref{apx:SSSBounds} we present some upper bounds on weighted sums of stable set sequences developed by Harvey and Vondrák~\cite{HarveyVondrak15}.
	In this subsection we show how to use these bounds for deducing upper bounds on the expected number of resamplings performed by Algorithm~\ref{alg:flawSelect}.
	
	\begin{definition}\label{def:StableSequence} (Stable set sequences)
		A sequence of sets ${\cal I}=(I_1,\ldots,I_s)$, such that $s\geq 0$, $\forall i \in [s]:\, I_i\in \Ind(F)\!\setminus\!\{\varnothing\}$, 
		and  $\forall r\in[s-1]:\, I_{r+1}\subseteq\Gamma^+(I_r)$, is called a {\em\bf stable set sequence}. For a stable set sequence ${\cal I}=(I_1,\ldots,I_s)$ we introduce notation $p_{\cal I}=\prod_{i=1}^{s}p_{I_i}$ and $\lVert {\cal I}\rVert=\sum_{i=1}^s|I_i|$.
		Also let us introduce ${\cal IS}:=\{{\cal I}: {\cal I} \text{ is a stable set sequence in } F\}$ for the {\em\bf set of stable set sequences}. 
		Finally let ${\cal IS}_T:=\{{\cal I}\in {\cal IS}: \lVert {\cal I}\rVert=T \}$,
		and ${\cal IS}_{\geq T}:=\{{\cal I}\in {\cal IS}: \lVert {\cal I}\rVert\geq T \}$.
	\end{definition}
	
	\begin{prop}
		The sequence of sets of flaws produced by a run of Algorithm~\ref{alg:flawSelect} is always a stable set sequence. (Note that we ignore the empty set produced by the final round.)
	\end{prop}
	\begin{proof}
		If no resamplings happen in a round, then Algorithm~\ref{alg:flawSelect} terminates and thus, apart from the final round, $I_i$ is non-empty. Let $U_i=F\setminus C$ denote the unchecked flaws at the end of the $i$-th round. Due to the condition in line~\ref{line:indepCondition}, $U_i\subseteq \Gamma^+(I_i)$. But observe that the flaws that become unchecked by a resampling during the $(i+1)$-st round cannot be addressed in this round, again by the condition in line~\ref{line:indepCondition}. So $I_{i+1}\subseteq U_i\subseteq \Gamma^+(I_i)$ as required.
	\end{proof}	
	
	\begin{prop} \label{prop:IndepAndLog}
		For every $L\in \mathcal{L}^B$ we can uniquely determine the stable-set sequence ${\cal I}_L\in{\cal IS}$ that is produced  by Algorithm~\ref{alg:flawSelect} after all measurement results in $L$ were processed. Moreover, this mapping is injective.
	\end{prop}
	\begin{proof}
		The unique determination of ${\cal I}_L$ follows from the deterministic nature of the classical part of Algorithm~\ref{alg:flawSelect}; the selection rule in line~\ref{line:deterministicSelection} plays a crucial role for establishing determinism.
		For injectivity note that for $L\in\mathcal{L}^{B(k)}$ we get $\left\lVert{\cal I}_L\right\rVert=k$. So it is enough to consider
		$L\neq L'\in\mathcal{L}^{B(k)}$ and show that ${\cal I}_L$ and ${\cal I}_{L'}$ differ.
		Since both $L$ and $L'$ are in $\mathcal{L}^{B(k)}$ it cannot happen that one is the continuation of the other.
		Let $i+1$ be the first position where they differ. Let $L_{i}\in\mathcal{L}_i$  denote the common part of $L$ and $L'$ before the $i+1$-st position.
		Without loss of generality, assume $L$ is a continuation of $(L_i,B)$ and $L'$ is the continuation of $(L_i,G)$. 
		Let us assume that after seeing $L_i$, the algorithm is in its $j$-th round and let $f_{L_i}$ be the next flaw it will address. Then $I_j$ (the $j$-th independent set in ${\cal I}_L$) contains $f_{L_i}$,
		whereas $I'_j$ (the $j$-th independent set in ${\cal I}_{L'}$) does not contain $f_{L_i}$, so ${\cal I}_L$ and ${\cal I}_{L'}$ are clearly different.
	\end{proof}

	Now we have all the tools for proving upper bounds on the expected number of resamplings that Algorithm~\ref{alg:flawSelect} perform.
	\begin{theorem} \label{thm:generalR}
		Let $\mathcal{Q}$ be a progressive quantum channel with respect to some subspace progress measure $\{V^C:C\subseteq F\}$ with error parameter $\theta \in [0,1]$. Let $\mathbb{E}[\#\mathrm{Resamplings}]$ denote the expected number of resamplings that Algorithm~\ref{alg:flawSelect} performs while using $\mathcal{Q}$ in line~\ref{line:applyQ}.
		If $(p_f)_{f\in F}\in(0,1)^{|F|}$ satisfies the condition
		\begin{align*}
		&\text{\textbullet } \text{ \eqref{eq:SLC}, then } 
		\,\mathbb{E}[\#\mathrm{Resamplings}] = \mathcal{O}\left(n\cdot\frac{|F|}{d}\right).  \\	
		&\text{\textbullet } \text{ \eqref{eq:ALC}, then } 
		\,\mathbb{E}[\#\mathrm{Resamplings}] = \mathcal{O}\left(n\cdot\sum_{f\in F} \frac{x_f}{1-x_f}\right).  \\
		&\text{\textbullet } \text{ \eqref{eq:CEC}, then } 
		\,\mathbb{E}[\#\mathrm{Resamplings}]  = \mathcal{O}\left(n\cdot\sum_{f\in F} y_f\right).  \\
		&\text{\textbullet } \text{ \eqref{eq:SHC}, then } 
		\,\mathbb{E}[\#\mathrm{Resamplings}]  = \mathcal{O}\left(n\cdot\sum_{f\in F} \frac{q_{\{f\}}}{q_{\varnothing}}\right). 	  
		\end{align*}
		Moreover if Algorithm~\ref{alg:flawSelect} terminates with ``SUCCESS", then the resulting quantum state lies in~$V^F$,
		and the probability of terminating with ``ERROR" is less than $2\theta \cdot R$.
	\end{theorem}
	\begin{proof}
		For the upper bound on $\mathbb{E}[\#\mathrm{Resamplings}]$ we invoke the results of Harvey and Vondrák~\cite{HarveyVondrak15} as presented in Corollary~\ref{cor:expected} in our Appendix~\ref{apx:SSSBounds}, using the following argument:
		
		\begin{align*}
		\mathbb{E}[\#\mathrm{Resamplings}]
		=& \sum_{k=1}^{\infty}k\cdot P\left(\text{doing exactly } k \text{ resamplings}\right) \\		
		=& \sum_{k=1}^{\infty}P\left(\text{doing at least } k \text{ resamplings}\right) & &(\text{see, e.g., \cite[Lemma 2.9]{mitzenmacher05probability} })\\	
		=&\sum_{k=1}^{\infty}\min\left(1,P\left(\text{doing at least } k \text{ resamplings}\right)\right)\\
		=&\sum_{k=1}^{\infty}\min\left(1,\sum_{L\in \mathcal{L}^{B(k)}}P\left(\text{seeing outcomes } L\right)\right)\\
		=& \sum_{k=1}^{\infty}\min\left(1,\sum_{L\in \mathcal{L}^{B(k)}}\tr{\rho_L}\right)\\
		\leq&\sum_{k=1}^{\infty}\min\left(1,\sum_{L\in \mathcal{L}^{B(k)}}p_L\right) & & (\text{by Lemma~\ref{lemma:NonComm}})\\
		\leq&\sum_{k=1}^{\infty}\min\left(1,\sum_{{\cal I}\in{\cal IS}_k} p_{\cal I}\right) & & (\text{by Proposition~\ref{prop:IndepAndLog}})\\	 
		\end{align*}	
		
		Now we apply Corollary~\ref{cor:expected} (from Appendix~\ref{apx:SSSBounds}), using the additional observation that $\log\left(\frac{1}{q_\varnothing}\right)=\mathcal{O}(n)$.
		To see this bound, first note, that since we work with qubits we have $\left(p_f\cdot 2^{|b(f)|}\right)\in \mathbb{N}$, and therefore by Definition~\ref{def:IndepPoly} $\left(q_\varnothing\cdot 2^{n}\right)\in \mathbb{Z}$.
		Due to Proposition~\ref{prop:ShearerImplied}, \eqref{eq:SHC} holds, but \eqref{eq:SHC} requires $q_\varnothing>0$. This implies $q_\varnothing\geq 2^{-n}$, and thus $\log\left(\frac{1}{q_\varnothing}\right)=\mathcal{O}(n)$. Finally in the case of \eqref{eq:SLC} we use the reduction of Proposition~\ref{prop:ShearerImplied} to the case \eqref{eq:ALC}.
		
		The statements about the quantum state at termination and the error probability follow from Lemma~\eqref{lemma:NonComm}.
	\end{proof}
	
	\begin{corollary}\label{cor:existential} (Existential Theorem)
		Suppose the projectors $\left(\Pi_f\right)_{f\in F}$ satisfy any of our four conditions \eqref{eq:SLC}-\eqref{eq:SHC}. Then the local Hamiltonian $H=\sum_{f\in F}\Pi_f$ is frustration-free, or in other words $\bigcap_{f\in F}\ker\left(\Pi_f\right)$ has dimension at least $1$.
	\end{corollary}
	\begin{proof}
		Proposition~\ref{prop:exactProgressive} shows that the exact quantum channel constructed in Definition~\ref{def:exactChannel} is a zero-error progressive quantum channel with respect to the exact progress measure. Moreover, Theorem~\ref{thm:generalR} shows that if run Algorithm~\ref{alg:flawSelect} using such a progressive quantum channel, then the resulting algorithm terminates with ``SUCCESS" with probability $1$ and at termination the quantum state lies in the kernel of $H$. Therefore, $\dim(\ker(H))=\dim\left(\bigcap_{f\in F}\ker\left(\Pi_f\right)\right)\geq 1$.
	\end{proof}
	
	Corollary~\ref{cor:existential} is almost equivalent to ~\cite[Corollary 1.6]{LovAmb} and~\cite[Theorem 1]{SattathLatice}; the only difference is, that the previous results provide a non-trivial lower-bound on $\dim(\ker(H))$, whereas our results only show that $\dim(\ker(H))\geq 1$.
	
	\subsection{Runtime bounds on different versions of the quantum algorithm}\label{subsec:diffVersions}
	
	\begin{corollary}
		If we run Algorithm~\ref{alg:flawSelect} with $\mathcal{Q}_f^C$ replaced by projective measurements $\Pi_f$, then the bounds on the expected number of resamplings from Theorem~\ref{thm:generalR} hold.
	\end{corollary}
	\begin{proof}
		Observe that the projective measurements $\Pi_f$ implement a zero-error (i.e., $\theta=0$) progressive quantum channel with respect to the trivial progress measure of Definition~\ref{def:trivialMeasure}.
	\end{proof}	
	
	\begin{corollary} \label{cor:commutingRuntime}
		If all the projectors $\{\Pi_f\}_{f\in F}$ commute pairwise and we run Algorithm~\ref{alg:flawSelect} with $\mathcal{Q}_f^C$ replaced by projective measurements $\Pi_f$, then the bounds on the expected number of resamplings from Theorem~\ref{thm:generalR} hold. Moreover, at termination the output state lies in the kernel of $H$.
	\end{corollary}
	\begin{proof}
		The statement follows from Proposition~\ref{prop:indeedGeneralisation}, Proposition~\ref{prop:exactProgressive} and Theorem~\ref{thm:generalR}.
	\end{proof}	
		
	\begin{corollary}\label{cor:approxRun}
		Let $R$ be the upper bound on the expected number of resamplings in Theorem~\ref{thm:generalR}, and let $d=\max\left\{|\Gamma^+(f)|:f\in F\right\}$.
		Suppose $\mathcal{Q}$ is a progressive quantum channel (Def.~\ref{def:progressive_channel}) with respect to the exact progress measure (Def.~\ref{def:exactMeasure}) with error parameter $\theta=\frac{1}{12R}$, and suppose $\tilde{\mathcal{Q}}$ is a $\frac{1}{6\left(|F|+6R\cdot d\right)}$-approximation (Def.~\ref{def:approximate_channel}) of $\mathcal{Q}$.
		Suppose we run Algorithm~\ref{alg:flawSelect} while using $\tilde{\mathcal{Q}}$ in line~\ref{line:applyQ}, and terminate it with ``TIMEOUT" if it attempts to do more than $6R$ resamplings. This quantum algorithm terminates with ``SUCCESS" with probability at least $1/2$, while using $\tilde{\mathcal{Q}}$ at most $(|F|+6R\cdot d)$ times in total. Conditional on termination with ``SUCCESS", its output quantum state $\tilde{\rho}$ is $1/2$-close in trace distance to a quantum state $\rho$ supported on the ground state space of $H$.
	\end{corollary}
	\begin{proof}
		First suppose we use the quantum channel $\mathcal{Q}$ in line~\ref{line:applyQ}.
		Theorem~\ref{thm:generalR} states that Algorithm~\ref{alg:flawSelect} terminates with ``ERROR" with probability at most $1/6$. Using Markov's inequality we can see, that the probability of termination with ``TIMEOUT" has probability at most $1/6$ also. 
		Let $\rho$ denote the output state when the algorithm uses $\mathcal{Q}$ and let $\rho_S$ denote the unnormalised output state corresponding to termination with ``SUCCESS", finally let $p_S=\tr{\rho_S}$ denote its probability. As we have show $p_S\geq 2/3$ and Theorem~\ref{thm:generalR} shows that $\rho_S$ is supported on the kernel of $H$.
		
		Let $T$ denote the number of applications of the quantum channel $\tilde{\mathcal{Q}}$.
		We claim that $T\leq |F|+6R\cdot d$. This can be seen by observing that each time $\tilde{\mathcal{Q}}$ returns with ``GOOD", $|C|$ is increased by one, and when it returns with ``BAD", i.e., a resampling occurs, $|C|$ decreases by at most $d-1$. Since we allow at most $6R$ resamplings  $|C|\geq T-d\cdot 6R$,  but $|F|\geq|C|$ proving the claim.
		
		Let $\tilde{\rho}$ denote the output state when the algorithm uses $\tilde{\mathcal{Q}}$ and let $\tilde{\rho}_S$ denote the unnormalised output state corresponding to termination with ``SUCCESS", finally let $\tilde{p}_S=\tr{\tilde{\rho}_S}$ denote its probability.	
		
		As we use $\tilde{\mathcal{Q}}$ at most $\left(|F|+6R\cdot d\right)$ times, and $\tilde{\mathcal{Q}}$ $\frac{1}{6\left(|F|+6R\cdot d\right)}$-approximates $\mathcal{Q}$, 
		by repeated use of the triangle inequality we can see that $\trnorm{\rho-\tilde{\rho}}\leq 1/6$, and so $|p_S-\tilde{p}_S|\leq\trnorm{\rho-\tilde{\rho}}\leq 1/6$. Also $\frac{\tilde{\rho}_S}{\tilde{p}_S}$, the output state conditioned on the ``SUCCESS" outcome, is $1/2$-close to $\frac{\rho_S}{p_S}$ in trace distance, as shown by the calculation below:
		\begin{align*}
		\trnorm{\frac{\rho_S}{p_S}-\frac{\tilde{\rho}_S}{\tilde{p}_S}} 
		=&\trnorm{\frac{\rho_S}{p_S}-\frac{\tilde{\rho}_S}{p_S}
			+\left(\frac{1}{p_S}-\frac{1}{\tilde{p}_S}\right)\tilde{\rho}_S}\\
		\leq& \trnorm{\frac{\rho_S}{p_S}-\frac{\tilde{\rho}_S}{p_S}}
		+\trnorm{\left(\frac{1}{p_S}-\frac{1}{\tilde{p}_S}\right)\tilde{\rho}_S}	\\	
		=& \frac{1}{p_S}\trnorm{\rho_S-\tilde{\rho}_S}
		+\frac{|\tilde{p}_S-p_S|}{\tilde{p}_S\cdot p_S}\tilde{p}_S	\\	
		\leq& \frac{2}{p_S}\trnorm{\rho-\tilde{\rho}}\\	
		\leq& \frac{1}{2}										
		\end{align*}
	\end{proof}		
	
	\subsection{Comparison with other LLL algorithms}\label{subsec:runtimeCompare}
	
	In the classical setting, the expected number of resamplings done by the Moser-Tardos algorithm, under the condition
	\begin{align*}
	&\text{\textbullet } \text{ \eqref{eq:SLC}, is } 
	\mathcal{O}\left(\frac{|F|}{d}\right)  \text{ as proven by Moser~\cite{MoserOrig}}.\footnote{}\\	
	&\text{\textbullet } \text{ \eqref{eq:ALC}, is } 
	\mathcal{O}\left(\sum_{f\in F} \frac{x_f}{1-x_f}\right)  \text{ as proven by Moser and Tardos~\cite{MoserTardos}.}\\
	&\text{\textbullet } \text{ \eqref{eq:SHC}, is } 
	\mathcal{O}\left(\sum_{f\in F} \frac{q_{\{f\}}}{q_{\varnothing}}\right) \text{ as proven by Kolipaka and Szegedy~\cite{KolipakaSzegedy}.}
	\end{align*} 
	\footnotetext{A weaker version was proven in this paper, however the result is implied by the follow-up work~\cite{MoserTardos}.}
	
	In all cases our bound for the quantum case in Theorem~\ref{thm:generalR} is worse by a linear factor in the number of qubits.
	This extra $n$ factor is a side effect of the ``forward-looking" analysis technique, and in the next version we will show how it can eliminated in the \eqref{eq:SLC} case, using the compression argument. We conjecture that this extra factor can be removed in the other cases as well.

	
	
	\section{Efficient implementation of the algorithm}\label{sec:efficient}
	
	\begin{definition} (Hamiltonians and their gap for subsystems)
		For $S\subseteq F$ let $H^S=\sum_{f'\in S}\Pi_{f'}$ denote the \textbf{Hamiltonian of the subsystem $S$}, 
		and let $\gamma^S$ be the smallest non-zero singular value of $H^S$ (or if $H^S=0$, then $\infty$). 
		Also let $\gamma=\min_{S \subseteq F}(\gamma^S)$ be the \textbf{\nom{uniform gap}{The smallest spectral gap of all sub-Hamiltonians}} of the system, denoted $\gamma$.
        \label{def:uniform_gap}
	\end{definition}

	We use the above definition for \nom{$\gamma$}{The uniform gap of the Hamiltonian} throughout this section. The QSAT instances that satisfy any of our four conditions \eqref{eq:SLC}-\eqref{eq:SHC} are frustration-free as shown by Corollary~\ref{cor:existential}. Since we only consider instances that satisfy some of these conditions, the two definitions for $\gamma$ (Definition~\ref{def:uniform_gap} and Eq.~\eqref{eq:def_gamma}) are equivalent for all the Hamiltonians in this work.

	\subsection{Non-commuting weak measurement procedure} \label{sec:weakAnalysis}
	The key idea for constructing a progressive quantum channel for non-commuting projectors is using the quantum Zeno effect, and performing several weak measurements of $\Pi_f$ with intensity~$\theta$, using the weak measurement operators
	\begin{equation}
	M_f^b=\sqrt{\theta}\Pi_f,\,\,M_f^g=(\Id-\Pi_f)+\sqrt{1-\theta}\Pi_f=\Id-(1-\sqrt{1-\theta})\Pi_f. \label{eq:weakM}
	\end{equation}
	Note that $M_f^b=(M_f^b)^\dagger$ and also $M_f^g=(M_f^g)^\dagger$.
	So $\{M_f^b, M_f^g\}$ are square roots of a POVM since $M_f^b\cdot (M_f^b)^\dagger+M_f^g\cdot (M_f^g)^\dagger=\Id$. For an intuitive explanation of these operators see the Introduction.
	
	The algorithm below implements a progressive quantum channel using some non-local measurements, and therefore it is not efficient in its present form.
	Later we show how to implement an approximate version of these non-local operations by using only local measurements in Algorithm~\ref{alg:proj}. The algorithm below can be interpreted as an approximate version of the exact quantum channel (Def.~\ref{def:exactChannel}), followed by a decoherence channel -- see Appendix~\ref{apx:Decoherence}.
	Also it can be understood from a geometric point of view using Jordan's Theorem -- see Appendix~\ref{apx:Jordan},
	\begin{algorithm}[H]
		\caption{$\mathbb{M}^\theta$ Non-commutative measurement for checked flaws $C\subseteq F$ and $f\in F$ }\label{alg:measurement}
		\begin{algorithmic}[1]
			\STATE {\bf input} quantum state $\rho$ and classical information $C\subseteq F$, $f\in F$ and $t\in \mathbb{N}$
			\STATE ($\star\,\,$ Let us write $V=\bigcap_{f'\in C}\ker(\Pi_{f'})$ and $V_f=V\cap\ker(\Pi_f)$  $\,\,\star$)
			\STATE {\bf repeat} $t$ times {\bf do}
			\STATE ~~~~ {\bf measure} $\{M_f^g,M_f^b\}$; {\bf if} result ``$b$" {\bf then return} ``B"
			\STATE ~~~~ {\bf measure} $\Pi_V$; {\bf if} result ``not in $V$" {\bf then return} ``E1"		
			\STATE {\bf end repeat}
			\STATE {\bf measure} $\Pi_{V_f}$; {\bf if} result ``not in $V_f$" {\bf then return} ``E2"		
			\STATE {\bf return} ``G"
		\end{algorithmic}
	\end{algorithm}
	
    The following lemma proves that Algorithm~\ref{alg:measurement} implements a progressive quantum channel. We prove slightly stronger properties than required by Definition ~\ref{def:progressive_channel}, as it fits the proof better.
    
	Let us fix $C\subseteq F$, $f\in F$ and $\theta \in (0,1]$, and use notation $V=\bigcap_{f'\in C}\ker(\Pi_{f'})$, $V_f=V\cap\ker(\Pi_f)$ and $V_{\overline{f}}=\bigcap_{f'\in C\setminus \Gamma^+(f)}\ker(\Pi_{f'})$.
	\begin{lemma}\label{lemma:MThProgr}
		For the input state $\rho\succeq 0$ let us denote the output of Algorithm~\ref{alg:measurement} by $\mathbb{M}^\theta(\rho)$.
		If $t\geq \frac{\ln(3/\theta)}{\theta\cdot\min\left(\gamma^{C\cup\{f\}},1\right)}$, then
		Algorithm~\ref{alg:measurement} implements a progressive quantum channel with respect to the exact progress measure and with error parameter $\theta$, as the following properties hold:\\
		(We only distinguish E1 and E2 for the sake of analysis, but we treat both of them just as E.)
		
		\begin{enumerate}[label=(\textbf{\roman*})]
			\item 
			$\mathbb{M}^\theta(\rho)=\mathbb{M}_G^\theta(\rho)\otimes G+
			\mathbb{M}_B^\theta(\rho)\otimes B+\mathbb{M}_E^\theta(\rho)\otimes E$. \label{it:outcomes}
			\item 
			$\mathbb{M}_G^\theta(\rho)=\Pi_{V_f}\rho\Pi_{V_f}$ \label{it:good}
			\item 
			$\mathbb{M}_B^\theta(\Pi_{V})=\mathbb{M}_B^\theta(\Pi_{V}-\Pi_{V_f})
			\preceq \Pi_{\im(\Pi_f\Pi_V)}\preceq \Pi_f\Pi_{V_{\overline{f}}}$ \label{it:bad}
			\item 
			$\,\tr{\mathbb{M}_E^\theta(\rho)}\leq 2\theta\cdot \tr{\mathbb{M}_B^\theta(\rho)}$ \label{it:error}
		\end{enumerate}
	\end{lemma}
	\begin{proof}
		First observe that the output state of Algorithm~\ref{alg:measurement} corresponding to the 4 possible outcome labels can be described as
		\begin{align}
		\mathbb{M}_G^\theta(\rho)=&\Pi_{V_f}\left(\Pi_V M_f^g\right)^t \rho \left( M_f^g\Pi_V\right)^t\Pi_{V_f} \label{eq:MG}\\
		\mathbb{M}_B^\theta(\rho)=& \sum_{i=0}^{t-1} M_f^b \left(\Pi_V M_f^g\right)^i \rho \left( M_f^g\Pi_V\right)^i M_f^b\label{eq:MB}\\
		\mathbb{M}_{E1}^\theta(\rho)=&\sum_{i=0}^{t-1} (\Id -\Pi_V) M_f^g \left(\Pi_V M_f^g\right)^i \rho \left( M_f^g\Pi_V\right)^i M_f^g (\Id -\Pi_V) \label{eq:ME1}\\
		\mathbb{M}_{E2}^\theta(\rho)=&(\Id-\Pi_{V_f})\left(\Pi_V M_f^g\right)^t \rho \left( M_f^g\Pi_V\right)^t(\Id-\Pi_{V_f})\label{eq:ME2}
		\end{align}
		
		\noindent Let us list some properties of the projection operators, which we will often use in the derivations:
		\begin{align}
		\Pi_f\Pi_{V_f}=& 0 & & (\text{since } V_f\subseteq \ker(\Pi_f)) \label{eq:fVfZero}\\
		M_f^g\Pi_{V_f}=& \Pi_{V_f} & & (\text{by } \eqref{eq:weakM},\eqref{eq:fVfZero}) \label{eq:MgVf}\\
		\Pi_V\Pi_{V_f}=& \Pi_{V_f} & & (\text{since } V_f\subseteq V) \label{eq:VfVComm}\\
		\Pi_f\Pi_{V_{\overline{f}}}=&\Pi_{V_{\overline{f}}}\Pi_f
		& & (\text{since they act on different qubits}) \label{eq:fVofComm}\\
		\Pi_f \Pi_V =& W\Sigma U^\dagger & & (\text{singular value decomposition})  \label{eq:fVfSVD}
		\end{align}
			
		\noindent \textbf{Property \ref{it:outcomes}:} For the sake of analysis we distinguish the two types of error outcomes of Algorithm~\ref{alg:measurement}, but otherwise we merge them:
		\begin{equation}
		\mathbb{M}_E^\theta(\rho)=\mathbb{M}_{E1}^\theta(\rho)+\mathbb{M}_{E2}^\theta(\rho). \label{eq:errors}
		\end{equation}	
		
		\noindent \textbf{Property \ref{it:good}:} 
		By using \eqref{eq:MgVf} and \eqref{eq:VfVComm} we get
		\begin{equation}
		\mathbb{M}_G^\theta(\rho)= \Pi_{V_f}\left(\Pi_V M_f^g\right)^t \rho \left( M_f^g\Pi_V\right)^t\Pi_{V_f}=\Pi_{V_f}\rho\Pi_{V_f} \label{eq:ii}
		\end{equation}
		
		\noindent \textbf{Property \ref{it:bad}:} 
		\begin{align*}
		\mathbb{M}_B^\theta(\Pi_V)
		=&\sum_{i=0}^{t-1} \theta \Pi_f \left(\Pi_V M_f^g\right)^i \Pi_V \left( M_f^g\Pi_V\right)^i \Pi_f
		& &  \kern-6mm(\text{by } \eqref{eq:weakM},\eqref{eq:MG})\\
		=&\sum_{i=0}^{t-1} \theta \Pi_f \left(\Pi_V \left(\Id-(1-\sqrt{1-\theta})\Pi_f\right)\right)^i 
		\Pi_V \left(\Pi_V \left(\Id-(1-\sqrt{1-\theta})\Pi_f\right)\right)^i \Pi_f & & (\text{by }\eqref{eq:weakM})\\
		=&\sum_{i=0}^{t-1} \theta \Pi_f \left(\Pi_V-(1-\sqrt{1-\theta})\Pi_V\Pi_f\right)^{i} \Pi_V 
		\Pi_V \left(\Pi_V-(1-\sqrt{1-\theta})\Pi_f\Pi_V\right)^{i} \Pi_f & & (\Pi_V=\Pi_V^2)\\
		=&\sum_{i=0}^{t-1} \theta \Pi_f \Pi_V\left(\Id-(1-\sqrt{1-\theta})\Pi_V\Pi_f\Pi_V\right)^{2i}\Pi_V\Pi_f\\
		=&\sum_{i=0}^{t-1} \theta W\Sigma U^\dagger \left(\Id-(1-\sqrt{1-\theta})U\Sigma W^\dagger W\Sigma U^\dagger \right)^{2i} U\Sigma W^\dagger & & (\text{by } \eqref{eq:fVfSVD})\\
		=&\sum_{i=0}^{t-1} \theta W\Sigma U^\dagger \left(U\left(\Id-(1-\sqrt{1-\theta})\Sigma^2\right) U^\dagger \right)^{2i} U\Sigma W^\dagger\\
		=&\theta W\Sigma\,\sgn{\Sigma} U^\dagger U\sum_{i=0}^{t-1} \left(\Id-(1-\sqrt{1-\theta})\Sigma^2\right)^{2i} U^\dagger  U\sgn{\Sigma}\Sigma W^\dagger
		& & \kern-10mm(\Sigma=\sgn{\Sigma}\Sigma)\\
		=&\theta W\Sigma\sum_{i=0}^{t-1} \left(\sgn{\Sigma}-(1-\sqrt{1-\theta})\Sigma^2\right)^{2i} \Sigma W^\dagger
		& & \kern-10mm(\Sigma=\sgn{\Sigma}\Sigma)\\
		\preceq& W\sgn{\Sigma} W^\dagger & &  \kern-7mm(\text{by Claim~\ref{claim:omomom}})\\
		=&\Pi_{\im(\Pi_f\Pi_V)} & & (\text{by } \eqref{eq:imAdj})\\ 
		\preceq&\Pi_f\Pi_{V_{\overline{f}}} & & (\text{by } \eqref{eq:imUpBnd})
		\end{align*}
		
		Finally it is easy to see from \eqref{eq:MB} that $\mathbb{M}_B^\theta(\Pi_{V_f})=0$, but we can also argue by trace preservation using \ref{it:good}: $$\tr{\mathbb{M}_B^\theta(\Pi_{V_f})}\leq \tr{\Pi_{V_f}}-\tr{\mathbb{M}_G^\theta(\Pi_{V_f})}=0.$$
		
		\begin{claim} \label{claim:omomom}
			For all $t\in\mathbb{N}^+$ and $\theta,\sigma \in (0,1]$ the following inequality holds: 
			\begin{equation}
			\sum_{i=0}^{t-1} \left(1-(1-\sqrt{1-\theta})\sigma^2\right)^{2i} \leq \frac{1}{\sigma^2\cdot \theta}. \label{eq:geomSum}
			\end{equation}
		\end{claim}
		\begin{proof}
			\begin{align*}
			\sum_{i=0}^{t-1} \left(1-(1-\sqrt{1-\theta})\sigma^2\right)^{2i}
			\leq \sum_{i=0}^{\infty} \left(1-(1-\sqrt{1-\theta})\sigma^2\right)^{2i} 
			=\frac{1}{1-\left(1-(1-\sqrt{1-\theta})\sigma^2\right)^2}  \leq \frac{1}{\sigma^2\cdot \theta} 
			\end{align*}
			The last inequality holds because:	
			\begin{align*}
			\frac{1}{1-\left(1-(1-\sqrt{1-\theta})\sigma^2\right)^2}  \leq& \frac{1}{\theta\cdot \sigma^2} \\
			\Updownarrow & \\
			\theta\cdot \sigma^2 \leq & 1-\left(1-(1-\sqrt{1-\theta})\sigma^2\right)^2 \\
			\Updownarrow & \\
			\theta\cdot \sigma^2 \leq & 2\sigma^2(1-\sqrt{1-\theta}) -  \sigma^4(1-\sqrt{1-\theta})^2 \\
			\Updownarrow & \\
			\theta\cdot \sigma^2 \leq & 2\sigma^2(1-\sqrt{1-\theta}) -  \sigma^4(1-2\sqrt{1-\theta} + 1-\theta) \\
			\Updownarrow & \\
			\theta\cdot \sigma^2(1-\sigma^2) -2 \sigma^2(1-\sigma^2) \leq &- 2\sigma^2(1-\sigma^2)\sqrt{1-\theta}\\
			\Uparrow & \\
			2-\theta\geq & 2\sqrt{1-\theta}\\
			\Updownarrow & \\
			(2-\theta)^2\geq & 4(1-\theta) 
			\end{align*}
		\end{proof}
		
		\noindent \textbf{Property \ref{it:error}:} 
		The following inequality is a manifestation of the quantum Zeno effect, and proves an error bound which is proportional to the intensity $\theta$ of the measurement. 
		\begin{align*}
		\mathbb{M}_{E1}^\theta(\Pi_V\rho\Pi_V)=&\sum_{i=0}^{t-1} (\Id -\Pi_V) M_f^g \left(\Pi_V M_f^g\right)^i \Pi_V\rho\Pi_V \left( M_f^g\Pi_V\right)^i M_f^g (\Id -\Pi_V)\\
		&\text{ Now we use } (\Id -\Pi_V)\cdot  \Id\cdot \Pi_V = 0 \text{, and subtract this zero operator: }  \\
		=& \sum_{i=0}^{t-1} (\Id -\Pi_V) (M_f^g-\Id) \left(\Pi_V M_f^g\right)^i \Pi_V\rho\Pi_V \left( M_f^g\Pi_V\right)^i (M_f^g-\Id) (\Id -\Pi_V)\\
		=& (1-\sqrt{1-\theta})^2\sum_{i=0}^{t-1} (\Id -\Pi_V) \Pi_f \left(\Pi_V M_f^g\right)^i \Pi_V\rho\Pi_V \left( M_f^g\Pi_V\right)^i \Pi_f (\Id -\Pi_V) & &  \kern-2mm(\text{by } \eqref{eq:weakM})\\
		\preceq& \theta^2\sum_{i=0}^{t-1} \Pi_f \left(\Pi_V M_f^g\right)^i \Pi_V\rho\Pi_V \left( M_f^g\Pi_V\right)^i \Pi_f  
		& & \kern-15mm(1-\sqrt{1-\theta}\leq \theta)\\
		=& \theta\cdot \mathbb{M}_B^{\theta}(\Pi_V\rho\Pi_V) & & \kern-10.5mm(\text{by } \eqref{eq:weakM},\eqref{eq:MB})
		\end{align*}
		
		\noindent Due to the assumption $\rho=\Pi_V\rho\Pi_V$, this result implies
		\begin{equation}
		\tr{\mathbb{M}_{E1}^\theta(\rho)}\leq\theta\cdot \tr{\mathbb{M}_B^{\theta}(\rho)}. \label{eq:trE1} 
		\end{equation}

		\noindent Now we start bounding the other type of error arising from doing few iterations (i.e., small $t$):
		\begin{align*}
		\mathbb{M}_{E2}^\theta(\Pi_V\rho\Pi_V)
		=&(\Id-\Pi_{V_f})\left(\Pi_V M_f^g\right)^t \Pi_V\rho\Pi_V \left( M_f^g\Pi_V\right)^t(\Id-\Pi_{V_f})& & (\text{by } \eqref{eq:ME2})\\
		=&(\Id-\Pi_{V_f})\Pi_V\left(\Pi_V M_f^g\right)^t \Pi_V\rho\Pi_V \left( M_f^g\Pi_V\right)^t \Pi_V(\Id-\Pi_{V_f})
		& & (\Pi_V=\Pi_V^2)\\
		=&(\Pi_V-\Pi_{V_f})\left(\Pi_V M_f^g\Pi_V\right)^{t}\Pi_V\rho\Pi_V\left(\Pi_V M_f^g\Pi_V\right)^{t} (\Pi_V-\Pi_{V_f})
		& & (\text{by } \eqref{eq:VfVComm})
		\end{align*}
		
		To continue we need to bound $(\Pi_V-\Pi_{V_f})\left(\Pi_V M_f^g\Pi_V\right)^{t}$. 
		For this let us define $\sigma$ as the smallest non-zero diagonal element of $\Sigma$ from \eqref{eq:fVfSVD}, or $1$ if $\Sigma=0$.
		
		\begin{align*}
		(\Pi_V-\Pi_{V_f})\left(\Pi_V M_f^g\Pi_V\right)^{t}
		=& (\Pi_V-\Pi_{V_f})\left(\Pi_V \left(\Id-\left(1-\sqrt{1-\theta}\right)\Pi_f\right)\Pi_V\right)^{t} & & (\text{by } \eqref{eq:weakM})\\
		=& (\Pi_V-\Pi_{V_f})\Pi_V \left(\Id-\left(1-\sqrt{1-\theta}\right)\Pi_V\Pi_f\Pi_V\right)^{t}\\
		=& (\Pi_V-\Pi_{V_f})\left(\Id-\left(1-\sqrt{1-\theta}\right)U\Sigma W^\dagger W\Sigma U^\dagger\right)^{t} & & \kern-4mm(\text{by } \eqref{eq:VfVComm},\eqref{eq:fVfSVD})\\
		=& U\text{sgn}(\Sigma) U^\dagger U\left(\Id-\left(1-\sqrt{1-\theta}\right)\Sigma^2 \right)^{t}U^\dagger & & (\text{by } \eqref{eq:badNonOrth})\\
		=& U\left(\text{sgn}(\Sigma)-\left(1-\sqrt{1-\theta}\right)\Sigma^2 \right)^{t}U^\dagger\\ 
		\preceq& U\left(\left(1-\sigma^2\theta/2\right)\text{sgn}(\Sigma) \right)^{t}U^\dagger
		& & \kern-15mm \left(1-\sqrt{1-\theta}\geq \theta/2\right)\\
		=&\left(1-\sigma^2\theta/2\right)^{t} (\Pi_V-\Pi_{V_f}) & & (\text{by } \eqref{eq:badNonOrth})\\
		\preceq&e^{-t\cdot\sigma^2\theta/2} (\Pi_V-\Pi_{V_f}) 
		\end{align*}
		
		\noindent so
		\begin{equation}
		\mathbb{M}_{E2}^\theta(\rho)\preceq e^{-t\cdot\sigma^2\theta} (\Pi_V-\Pi_{V_f})\Pi_V\rho\Pi_V (\Pi_V-\Pi_{V_f})=e^{-t\cdot\sigma^2\theta} (\Pi_V-\Pi_{V_f})\rho (\Pi_V-\Pi_{V_f}).  \label{eq:E2tBound}
		\end{equation}
		
		\noindent Now we bound $\sigma^2$ from below. Let $H_f=H^{C\cup\{f\}}$ 
		and $\gamma_{f}=\gamma^{C\cup\{f\}}$. Observe $V_f=\ker(H_f)$, so
		\begin{equation}
		\min(\gamma_{f},1)\cdot\Id \preceq H_f+\Pi_{V_f}.  \label{eq:Hf}
		\end{equation}
		We prove $\min(\gamma_{f},1)\leq \sigma^2$, by
		\begin{align*}
		U\cdot\min(\gamma_{f},1)\cdot\text{sgn}(\Sigma) U^\dagger =&
		\min(\gamma_{f},1)\cdot(\Pi_V-\Pi_{V_f}) & & (\text{by } \eqref{eq:badNonOrth})\\
		=&(\Pi_V-\Pi_{V_f})\cdot\min(\gamma_{f},1)\Id\cdot(\Pi_V-\Pi_{V_f}) \\
		\preceq& (\Pi_V-\Pi_{V_f})(H_f+\Pi_{V_f})(\Pi_V-\Pi_{V_f}) & & (\text{by } \eqref{eq:Hf})\\
		=& (\Pi_V-\Pi_{V_f})H_f(\Pi_V-\Pi_{V_f}) & & (\text{by } \eqref{eq:VfVComm})\\
		=&\Pi_V H_f \Pi_V & & (V_f=\ker(H_f))\\
		=&\Pi_V( H^{C} +\Pi_f)\Pi_V & & (\text{by definition})\\
		=&\Pi_V \Pi_f \Pi_V & & (V=\ker(H^{C}))\\
		=&U \Sigma^2 U^\dagger. & & (\text{by } \eqref{eq:fVfSVD})
		\end{align*}
		Since $t\geq \frac{\ln(3/\theta)}{\theta\cdot\min(\gamma_{f},1)}$ and $\min(\gamma_{f},1)\leq \sigma^2$, 
		we have $t\geq \frac{\ln(3/\theta)}{\theta\cdot\sigma^2	}$, and thus by \eqref{eq:E2tBound}
		\begin{equation}
		\mathbb{M}_{E2}^\theta(\rho)\preceq \frac{\theta}{3} (\Pi_V-\Pi_{V_f})\rho (\Pi_V-\Pi_{V_f}).  \label{eq:E2thBound}
		\end{equation}
		
		Note that since $\mathbb{M}^\theta$ is a quantum channel, it is trace-preserving: $\Tr[\rho]=\Tr[\mathbb{M}^\theta(\rho)]$.
		\begin{align*}
        -\Tr[\mathbb{M}^\theta(\rho)]=&-\Tr[\rho]\\
		\Tr[(\Id-\Pi_{V_f})\rho(\Id-\Pi_{V_f})]=&\Tr[\rho]-\Tr[\Pi_{V_f}\rho\Pi_{V_f}]\\
		\Tr[\mathbb{M}^\theta(\rho)]-\Tr[\mathbb{M}_G^\theta(\rho)]
		=&\Tr[\mathbb{M}_B^\theta(\rho)]+\Tr[\mathbb{M}_E^\theta(\rho)] & & \kern-10mm(\text{by } \ref{it:outcomes}) \\
		\Tr[\mathbb{M}_G^\theta(\rho)]=&\Tr[\Pi_{V_f}\rho\Pi_{V_f}] & & \kern-10mm(\text{by } \ref{it:good}) \\
		\Downarrow& \text{ sum up the above 4 equalities}\\
		\Tr[(\Id-\Pi_{V_f})\rho(\Id-\Pi_{V_f})]=&\Tr[\mathbb{M}_B^\theta(\rho)]+\Tr[\mathbb{M}_E^\theta(\rho)]\\
		\Downarrow& \text{ assume } \rho=\Pi_V\rho\Pi_V \text{ and use \eqref{eq:VfVComm}}\\
		\Tr[(\Pi_V-\Pi_{V_f})\rho(\Pi_V-\Pi_{V_f})]=&\Tr[\mathbb{M}_B^\theta(\rho)]+\Tr[\mathbb{M}_E^\theta(\rho)]\\
		\Updownarrow& \text{ using } \eqref{eq:errors}\\
		\underset{a}{\underbrace{\tr{(\Pi_V-\Pi_{V_f})\rho(\Pi_V-\Pi_{V_f})}}}=&
		\underset{b}{\underbrace{\tr{\mathbb{M}_B^\theta(\rho)}}}
		+\underset{c}{\underbrace{\tr{\mathbb{M}_{E1}^\theta(\rho)}}}
		+\underset{d}{\underbrace{\tr{\mathbb{M}_{E2}^\theta(\rho)}}}
		\end{align*}
		
		\noindent Equation \eqref{eq:trE1} and \eqref{eq:E2thBound} show that with the above choice of $a,b,c,d$, the conditions of Claim~\ref{claim:thetaRatio} hold and so
		$$\tr{\mathbb{M}_{E}^\theta(\rho)}
		=\tr{\mathbb{M}_{E1}^\theta(\rho)}+\tr{\mathbb{M}_{E2}^\theta(\rho)}
		\leq 2\theta \cdot\Tr[\mathbb{M}_B^\theta(\rho)].$$

		\begin{claim} \label{claim:thetaRatio}
			$\forall a,b,c,d\in\mathbb{R}_0^+$, if $d\leq \frac{\theta}{3}\cdot a$, $c\leq \theta\cdot b$ and $a=b+c+d$, then $(c+d)\leq 2\theta\cdot b$. 
		\end{claim}
		\begin{proof} Note $\theta\in (0,1]$, so if $a=0$, then $d\leq \theta\cdot b$ trivially holds. Else if $a>0$
			\begin{align*}
			\left(1-\theta/3\right)\cdot a \leq a-d =& b+c \leq (1+\theta)\cdot b \\
			\Downarrow & \\
			\frac{\left(1-\theta/3\right)}{(1+\theta)}\leq& b/a\\
			\Downarrow & \\
            d = \frac{d}{a}\cdot\frac{a}{b}\cdot b
            \leq \frac{\theta}{3}\cdot\frac{(1+\theta)}{\left(1-\theta/3\right)}\cdot b
			\leq& \frac{\theta}{3}\cdot\frac{2}{2/3}\cdot b=\theta b
			\end{align*}
		\end{proof}
		This completes the proof that properties \ref{it:outcomes}-\ref{it:error} hold.
	\end{proof} 
	\newpage
	
	\subsection{Approximate kernel projection procedure}\label{sec:kerProj}
	
	Let $S\subseteq F$, and $V=V^S=\ker(H^S)=\bigcap_{f\in S}\ker(\Pi_f)$.
	For the implementation of Algorithm~\ref{alg:measurement} we need the non-local measurement operator $\Pi_V$, but it turns out that for the purposes of the algorithm it is enough to implement a ``destructive" version of this measurement operator which we denote by $\tilde{\Pi}^S$.
	
	\begin{definition} \label{def:DP} (Destructive non-local measurement channel)
		For $S\subseteq F$ let $\tilde{\Pi}^S$ denote the quantum channel which performs the projective measurement on $V=V^S$, followed by a completely depolarising channel conditioned on the 
		$V^\perp$ outcome, and which labels its outputs with classical labels $(P,D)$ corresponding to 
		(``Projected",``Depolarised -- was not in $V$"). 
		Formally
		$\tilde{\Pi}^S:\mathbb{C}^{N\times N}\rightarrow\mathbb{C}^{N\times N}\otimes \mathbb{R}^2$,
		such that for input $0\preceq\rho \in \mathbb{C}^{N\times N}$ the output of the channel is 
		$$\tilde{\Pi}^S(\rho)=\tilde{\Pi}^S_P(\rho)\otimes P +\tilde{\Pi}^S_D(\rho)\otimes D,$$
		where
		$$\tilde{\Pi}^S_P(\rho)=\Pi_V\rho\Pi_V$$
		and 
		$$\tilde{\Pi}^S_D(\rho)=\tr{(\Id-\Pi_V)\rho(\Id-\Pi_V)}\cdot\Id/N.$$	 
	\end{definition}
	
	We implement an approximate version of the above channel via the following algorithm:
	
	\begin{algorithm}[H]
		\caption{$\tilde{\Pi}^{S,\tau}$ - Non-commutative kernel projection}\label{alg:proj}
		\begin{algorithmic}[1]
			\STATE {\bf input} quantum state $\rho$
			\STATE {\bf repeat} $\tau$ times {\bf do}
			\STATE ~~~~ choose $f\in S$ uniformly at random
			\STATE ~~~~ {\bf measure} $\Pi_f$
			\STATE ~~~~ {\bf if} result ``flaw $f$ is present" {\bf then} 
			\STATE ~~~~~~~~ completely depolarise qubits
			($\star\,\,$ i.e., resample all qubits  $\,\,\star$) \label{line:depolarisation}	
			\STATE ~~~~~~~~ {\bf return} ``$D$"				
			\STATE {\bf end repeat}	
			\STATE {\bf return} ``$P$"
		\end{algorithmic}
	\end{algorithm}
	\noindent Note that when Algorithm~\ref{alg:measurement} uses this algorithm as a subroutine, it throws away the output if it is labelled by $D$, since it indicates measurement outcome $V^\perp$. Therefore later we can safely ignore the depolarisation step in line~\ref{line:depolarisation}, as it is added just for the sake of analysis.

	In the proof of the following lemma we are going to use a special case of Hölder's inequality:\kern-1mm
	\begin{prop} \label{prop:HolderIneq}
		$\trnorm{AB}\leq\trnorm{A}\cdot\spnorm{B}$, where $\spnorm{B}$ denotes the spectral norm of $B$.
	\end{prop}
	\begin{proof}
		This is a special case of Hölder's inequality for Schatten $p$-norms~\cite[Cor. IV.2.6]{BhatiaMatrixAnal97}.
	\end{proof}
	
	\begin{lemma}\label{lemma:projPrec}
		For $S\subseteq F$, $\tau\in\mathbb{N}$ and $0\preceq\rho\in \mathbb{C}^{N\times N}$ let
		$\tilde{\Pi}^{S,\tau}(\rho)$ denote the output state of Algorithm~\ref{alg:proj}. Then its trace distance from the output of the ideal quantum channel $\tilde{\Pi}^S(\rho)$ can be bounded as follows:
		$$
		\left\lVert\tilde{\Pi}^{S,\tau}(\rho)-\tilde{\Pi}^S(\rho)\right\rVert_1
		\leq 4\cdot \exp\left(- \frac{\gamma^S}{|S|} \tau\right) \tr{\rho}.
		$$
		
	\end{lemma}
	\begin{proof}
		Let us first examine the case $\tau=1$, with outcome $P$:
		\begin{equation}
		\tilde{\Pi}^{S,1}_P(\rho)=\frac{1}{|S|}\sum_{f\in S}(\Id-\Pi_f)\rho(\Id-\Pi_f).\label{eq:PFDef}
		\end{equation}
		
		\noindent Observe, that by definition for all $f\in S$ we have $\ker(\Pi_f)\subseteq V$, thus $\Pi_V\Pi_f=0$ and so 
		\begin{equation}
		\Pi_V\tilde{\Pi}^{S,1}_P(\rho)\Pi_V=\frac{1}{|S|}\sum_{f\in S}(\Id-\Pi_f)\Pi_V\rho\Pi_V(\Id-\Pi_f)=\Pi_V\rho\Pi_V.\label{eq:VReam1}
		\end{equation}
		
		\noindent Also observe that due to the iterative structure of Algorithm~\ref{alg:proj} we have 
		\begin{equation}
		\tilde{\Pi}^{S,\tau+1}_P(\rho)=\tilde{\Pi}^{S,1}_P(\tilde{\Pi}^{S,\tau}_P(\rho)).\label{eq:recursion}
		\end{equation}
		
		\noindent Combining \eqref{eq:VReam1} and \eqref{eq:recursion} using induction yields
		\begin{equation}
		\Pi_V\tilde{\Pi}^{S,\tau}_P(\rho)\Pi_V=\Pi_V\rho\Pi_V.\label{eq:VReam}
		\end{equation}
		
		Now let us turn to the case $\tau=1$, with outcome $D$:
		
		\begin{align*}
		\tr{\tilde{\Pi}^{S,1}_D(\rho)}
		=&\tr{\frac{1}{|S|}\sum_{f\in S}\Pi_f\rho \Pi_f}\\
		=&\frac{1}{|S|}\sum_{f\in S}\Tr\left[\Pi^2_f\rho \right] & & (\Tr[AB]=\Tr[BA])\\
		=&\frac{1}{|S|}\Tr\left[H^{S}\rho \right] & & (\text{by definition})\\
		=&\frac{1}{|S|}\Tr\left[(\Id-\Pi_V)H^{S}(\Id-\Pi_V)\rho \right] & & (V=\ker{H^S})\\
		=&\frac{1}{|S|}\Tr\left[H^{S}(\Id-\Pi_V)\rho(\Id-\Pi_V)\right] & & (\Tr[AB]=\Tr[BA])\\
		=&\frac{1}{|S|}\Tr\left[\left(H^{S}+\gamma^{S}\Pi_V\right)(\Id-\Pi_V)\rho(\Id-\Pi_V)\right]
		& & (\Pi_V(\Id-\Pi_V)=0)\\
		\geq&\frac{1}{|S|}\Tr\left[\left(\gamma^{S}\Id\right)(\Id-\Pi_V)\rho(\Id-\Pi_V)\right]
		& & (\gamma^S\Id\preceq H^{S}+\gamma^{S}\Pi_V)\\
		= &\frac{\gamma^{S}}{|S|}\Tr\left[(\Id-\Pi_V)\rho(\Id-\Pi_V)\right].
		\end{align*}
		Combining this inequality with a trace-preservation argument yields 
		\begin{align}
		\tr{\tilde{\Pi}^{S,1}_P((\Id-\Pi_V)\rho(\Id-\Pi_V))}
		=&\Tr\left[(\Id-\Pi_V)\rho(\Id-\Pi_V)\right]
		-\tr{\tilde{\Pi}^{S,1}_D((\Id-\Pi_V)\rho(\Id-\Pi_V))}\nonumber\\
		\leq& \left(1-\frac{\gamma^{S}}{|S|}\right)\cdot \Tr\left[(\Id-\Pi_V)\rho(\Id-\Pi_V)\right]. \label{eq:dissUBound}
		\end{align}
		
		From the form of equation \eqref{eq:PFDef} we can see that 
		\begin{align}
		\Pi_V\tilde{\Pi}^{S,1}_P(\rho)=&\tilde{\Pi}^{S,1}_P(\Pi_V\rho)\nonumber\\
		\Downarrow& \left(\text{by induction, using }\eqref{eq:recursion}\right)\nonumber\\
		\Pi_V\tilde{\Pi}^{S,\tau}_P(\rho)=&\tilde{\Pi}^{S,\tau}_P(\Pi_V\rho)\label{eq:lOut}\\
		\Downarrow& (\text{by taking the adjoint of both sides})\nonumber\\
		\tilde{\Pi}^{S,\tau}_P(\rho)\Pi_V=&\tilde{\Pi}^{S,\tau}_P(\rho\Pi_V)\\
		\Downarrow& (\text{by linearity})\nonumber\\
		\tilde{\Pi}^{S,\tau}_P(\rho)(\Id-\Pi_V)=&\tilde{\Pi}^{S,\tau}_P(\rho(\Id-\Pi_V))
		\label{eq:lIn}\\
		\Downarrow& (\text{by taking the adjoint and by linearity})\nonumber\\
		(\Id-\Pi_V)\tilde{\Pi}^{S,\tau}_P(\rho)(\Id-\Pi_V)=&\tilde{\Pi}^{S,\tau}_P((\Id-\Pi_V)\rho(\Id-\Pi_V)) \label{eq:rInlIn}
		\end{align}
		Using equation \eqref{eq:recursion} recursively and applying \eqref{eq:rInlIn} together with \eqref{eq:dissUBound}, results in 
		\begin{equation}
		\Tr\left[(\Id-\Pi_V)\tilde{\Pi}^{S,\tau}_P(\rho)(\Id-\Pi_V)\right]\leq \left(1-\frac{\gamma^{S}}{|S|}\right)^\tau\cdot \Tr\left[(\Id-\Pi_V)\rho(\Id-\Pi_V)\right].\label{eq:disstUBound}
		\end{equation}

		\begin{claim}
			\begin{equation}
			\Pi_V\tilde{\Pi}^{S,\tau}_P(\rho)(\Id-\Pi_V)=\Pi_V\rho(\Id-\Pi_V)\left(\Id-\Pi_V-H^{S}/|S|\right)^\tau\label{eq:cohBound}
			\end{equation}
		\end{claim}
		\begin{proof}
			For $\tau=0$ this definitely holds, we proceed by induction. Suppose the statement holds for $\rho_\tau=\tilde{\Pi}^{S,\tau}_P(\rho)$, i.e, $\Pi_V\rho_\tau(\Id-\Pi_V)=\Pi_V\rho(\Id-\Pi_V)\left(\Id-\Pi_V - H^{S}/|S|\right)^\tau$, then we show it for $\tau+1$:
			\begin{align*}
			\Pi_V\tilde{\Pi}^{S,\tau+1}_P(\rho)(\Id-\Pi_V)
			=&\Pi_V\tilde{\Pi}^{S,1	}_P(\rho_\tau)(\Id-\Pi_V)  & &  (\text{by }\eqref{eq:recursion})\\
			=&\Pi_V\frac{1}{|S|}\sum_{f\in S}(\Id-\Pi_f)\rho_\tau (\Id-\Pi_f)(\Id-\Pi_V)  
			& &  (\text{by }\eqref{eq:PFDef})\\
			=&\frac{1}{|S|}\sum_{f\in S}\Pi_V\rho_\tau(\Id-\Pi_V) (\Id-\Pi_f)
			& &  (V\subseteq \ker(\Pi_f))\\
			=&\Pi_V\rho_\tau(\Id-\Pi_V)\frac{1}{|S|}\sum_{f\in S}(\Id-\Pi_f)\\
			=&\Pi_V\rho_\tau(\Id-\Pi_V)\left(\Id-H^{S}/|S|\right) 
			& &  (\text{by definition of }H^S)\\
			=&\Pi_V\rho_\tau(\Id-\Pi_V)\left(\Id-\Pi_V-H^{S}/|S|\right) 
			& &  ((\Id-\Pi_V)\Pi_V=0)\\
			=&\Pi_V\rho(\Id-\Pi_V)\left(\Id-\Pi_V-H^{S}/|S|\right)^{\tau+1}
			& &  (\text{by the ind. hyp.})
			\end{align*}
		\end{proof}
		
		Now we are ready to calculate the trace distance between $\tilde{\Pi}^{S}(\rho)$ and its approximation $\tilde{\Pi}^{S,\tau}(\rho)$.
		Since $(P,D)$ are classical labels
		
		\begin{equation}
		\trnorm{\tilde{\Pi}^{S,\tau}(\rho)-\tilde{\Pi}^{S}(\rho)}
		=\trnorm{\tilde{\Pi}^{S,\tau}_P(\rho)-\tilde{\Pi}^{S}_P(\rho)}
		+\trnorm{\tilde{\Pi}^{S,\tau}_D(\rho)-\tilde{\Pi}^{S}_D(\rho)}.\label{eq:trClSplit}
		\end{equation}
		
		\noindent We first handle the outcome $P$: (recall that we assumed $0\preceq \rho$ )
		\begin{align*}
		&\trnorm{\tilde{\Pi}^{S,\tau}_P(\rho)-\tilde{\Pi}^{S}_P(\rho)}
		=\trnorm{\tilde{\Pi}^{S,\tau}_P(\rho)-\Pi_V\rho\Pi_V} & & (\text{by Def.~\eqref{def:DP}}) \\
		&=\trnorm{\left(\left(\Id-\Pi_V\right)+\Pi_V\right)
			\left(\tilde{\Pi}^{S,\tau}_P(\rho)-\Pi_V\rho\Pi_V\right)
			\left(\left(\Id-\Pi_V\right)+\Pi_V\right)} & & (\Id=\left(\Id-\Pi_V\right)+\Pi_V) \\
		&\leq\trnorm{\Pi_V\left(\tilde{\Pi}^{S,\tau}_P(\rho)-\rho\right)\Pi_V}
		+2\cdot\trnorm{\Pi_V\tilde{\Pi}^{S,\tau}_P(\rho)\left(\Id-\Pi_V\right)} 
		& & (\trnorm{A}=\lVert A^\dagger\rVert_1) \\
		&\phantom{--}+\trnorm{\left(\Id-\Pi_V\right)\tilde{\Pi}^{S,\tau}_P(\rho)\left(\Id-\Pi_V\right)}
		& & (\text{triangle inequality}) \\
		&=2\cdot\trnorm{\Pi_V\rho\left(\Id-\Pi_V\right)\left(\Id-\Pi_V-H^{S}/|S|\right)^\tau}
		& & (\text{by }\eqref{eq:VReam}\text{ and }\eqref{eq:cohBound}) \\
		&\phantom{--}+\Tr\left[\left(\Id-\Pi_V\right)\tilde{\Pi}^{S,\tau}_P(\rho)\left(\Id-\Pi_V\right)\right] & & \left(0\preceq \tilde{\Pi}^{S,\tau}_P(\rho)\right)\\
		&\leq 2\cdot\trnorm{\Pi_V\rho\left(\Id-\Pi_V\right)}
		\cdot\spnorm{\left(\Id-\Pi_V-H^{S}/|S|\right)^\tau}
		& & (\text{Proposition~\ref{prop:HolderIneq}}) \\
		\phantom{-}&+\left(1-\frac{\gamma^{S}}{|S|}\right)^\tau\cdot \Tr\left[(\Id-\Pi_V)\rho(\Id-\Pi_V)\right] & &  (\text{by }\eqref{eq:disstUBound}) \\
		&=\left(1-\frac{\gamma^{S}}{|S|}\right)^\tau
		\left(2\cdot\trnorm{\Pi_V\rho\left(\Id-\Pi_V\right)}
		+\trnorm{(\Id-\Pi_V)\rho(\Id-\Pi_V)}\right)
		& & \left(\frac{\gamma^S}{|S|}\cdot\Id\preceq \frac{H^S}{|S|}+\Pi_V\right)\\
		&\leq\left(1-\frac{\gamma^{S}}{|S|}\right)^\tau\cdot 3\cdot\trnorm{\rho}
		& & (\text{Proposition~\ref{prop:HolderIneq}}) \\
		&\leq 3\cdot e^{-\frac{\gamma^S}{|S|}\tau  } \tr{\rho}. & & (1-x\leq e^{-x})
		\end{align*}
		
		\noindent Now we handle the outcome $D$: 
		
		\begin{align*}
		\tr{\tilde{\Pi}^{S,\tau}_D(\rho)}
		=&\tr{\rho}-\tr{\tilde{\Pi}^{S,\tau}_P(\rho)} & & (\text{trace preservation}) \\
		=&\tr{\Pi_V\rho\Pi_V} +\tr{(\Id-\Pi_V)\rho(\Id-\Pi_V)} & & (\text{trace preservation}) \\
		&\phantom{-}-
		\tr{\Pi_V\tilde{\Pi}^{S,\tau}_P(\rho)\Pi_V} -\tr{(\Id-\Pi_V)\tilde{\Pi}^{S,\tau}_P(\rho)(\Id-\Pi_V)} & & (\text{trace preservation}) \\
		=&\tr{(\Id-\Pi_V)\rho(\Id-\Pi_V)} -\tr{(\Id-\Pi_V)\tilde{\Pi}^{S,\tau}_P(\rho)(\Id-\Pi_V)} & & (\text{by \eqref{eq:VReam}})
		\end{align*}
		Using \eqref{eq:disstUBound} we can conclude that 
		\begin{equation} \label{eq:polDiff}
		\left|\tr{\tilde{\Pi}^{S,\tau}_D(\rho)}- \tr{(\Id-\Pi_V)\rho(\Id-\Pi_V)}\right|
		\leq \left(1-\frac{\gamma^{S}}{|S|}\right)^\tau \tr{(\Id-\Pi_V)\rho(\Id-\Pi_V)}.
		\end{equation}
		Finally, we use that both channels depolarise their output $D$: 
		\begin{align*}
		\trnorm{\tilde{\Pi}^{S,\tau}_D(\rho)-\tilde{\Pi}^{S}_D(\rho)}
		=&\trnorm{\tr{\tilde{\Pi}^{S,\tau}_D(\rho)}\cdot\Id/N
			-\tr{\tilde{\Pi}^{S}_D(\rho)}\cdot\Id/N} & & (\text{depolarised output}) \\
		=&\left|\tr{\tilde{\Pi}^{S,\tau}_D(\rho)}
		-\tr{\tilde{\Pi}^{S}_D(\rho)}\right|\cdot \trnorm{\Id/N}\\
		=&\left|\tr{\tilde{\Pi}^{S,\tau}_D(\rho)} 
		-\tr{(\Id-\Pi_V)\rho(\Id-\Pi_V)}\right| & & (\text{by Def.~\ref{def:DP}}) \\
		\leq&\left(1-\frac{\gamma^{S}}{|S|}\right)^\tau \tr{(\Id-\Pi_V)\rho(\Id-\Pi_V)}
		& & (\text{by \eqref{eq:polDiff}}) \\
		\leq& e^{-\frac{\gamma^S}{|S|} \tau} \tr{\rho}.
		\end{align*}
		
		Combining the bounds that we got for output cases $P$ and $D$ and using \eqref{eq:trClSplit} we can conclude that for $\rho\succeq 0$:
		$$
		\trnorm{\tilde{\Pi}^{S,\tau}(\rho)-\tilde{\Pi}^{S}(\rho)} \leq 4\cdot e^{- \frac{\gamma^S}{|S|} \tau} \tr{\rho}
		$$
	\end{proof}
	
	\subsection{The final algorithm}
	
	\begin{lemma} \label{lemma:QImplementation} (A progressive quantum channel implementation)
		Suppose $\theta,\beta\in(0,1]$. Let $t\geq\frac{\ln(3/\theta)}{\theta\cdot\gamma}$ and 
		$\tau\geq\frac{|F|}{\gamma}\left(\ln\left(1/\beta\right)+\ln\left(t+1\right)+\ln(4)\right)$.
		Consider Algorithm~\ref{alg:measurement} with parameters $\theta$ and $t$ 
		while replacing the non-local measurement operators with Algorithm~\ref{alg:proj} setting runtime to $\tau$.
		Then this algorithm implements a quantum channel that $\beta$-approximates a progressive quantum channel with respect to the exact progress measure with error parameter $\theta$. Moreover the algorithm uses at most $t$  local weak measurements and $(t+1)\tau$ local (strong) measurements.
	\end{lemma}
	\begin{proof}
		Algorithm~\ref{alg:measurement} with the above parameters implements a progressive quantum channel with respect to the exact progress measure with error parameter $\theta$ if we use the destructive non-local measurement channels of Definition~\ref{def:DP}, as shown by Lemma~\ref{lemma:MThProgr}. 
		Also Algorithm~\ref{alg:proj} implements a $\frac{\beta}{(t+1)}$-approximation of the destructive non-local measurement channels, as shown by Lemma~\ref{lemma:projPrec}. Since Algorithm~\ref{alg:measurement} uses these channel at most $t+1$ times, the triangle inequality shows that the approximate algorithm is $\beta$-close to the ideal one using the exact version of the destructive non-local measurement channel.
		The upper bound on the number of measurements performed can be easily deduced from the loop structure of 
		Algorithm~\ref{alg:measurement} and Algorithm~\ref{alg:proj}.
	\end{proof}
	
	\begin{theorem}\label{thm:approxRun}
		Let $R$ be the upper bound on the resamplings in Theorem~\ref{thm:generalR}, and let $d=\max\left\{|\Gamma^+(f)|:f\in F\right\}$. Let us run the algorithm of Corollary~\ref{cor:approxRun}
		with the quantum channel provided by Lemma~\ref{lemma:QImplementation} 
		using parameters $\theta=\frac{1}{12R}$ and $\beta=\frac{1}{6\left(|F|+6R\cdot d\right)}$.
		Then this quantum algorithm terminates with ``SUCCESS" with probability at least $1/2$, 
		while performing at most $\mathcal{O}\left(\frac{R\cdot|F|\cdot(|F|+R\cdot d)\cdot\log(R)}{\gamma^2}\cdot
		\left(\log(|F|)+\log(R)+\log\left(1/\gamma\right)\right)\right)$ (weak and strong) measurements. Conditional on termination with ``SUCCESS", its output quantum state is $1/2$-close in trace distance to a quantum state supported on the ground state space of $H$.
	\end{theorem}
	\begin{proof}
		The statement directly follows from Corollary~\ref{cor:approxRun} and Lemma~\ref{lemma:QImplementation}.
		To justify the bound on the number of measurements performed we note that $t=\frac{12R\cdot\ln(36 R)}{\gamma}$,
		$\tau=\frac{|F|}{\gamma}\left(\ln\left(|F|+6R\cdot d\right)+\ln\left(t+1\right)+\ln(4\cdot 6)\right)$
		and the algorithm of Lemma~\ref{lemma:QImplementation} is used at most $\left(|F|+6R\cdot d\right)$ times.
	
		The overall number of measurements performed by the algorithm is upper bounded by 
		$\left(|F|+6d\cdot R\right)\cdot (t+1)\cdot(\tau+1)=
		\mathcal{O}\left(\frac{R\cdot|F|\cdot(|F|+R\cdot d)\cdot\log(R)}{\gamma^2}\cdot
		\left(\log(|F|)+\log(R)+\log\left(1/\gamma\right)\right)\right)$.
	\end{proof}
	\begin{corollary} \label{cor:boosted}
		Let $R$ be the upper bound on the resamplings in Theorem~\ref{thm:generalR}, 
		and let $d=\max\left\{|\Gamma^+(f)|:f\in F\right\}$.
		Then for all $\delta,\epsilon\in(0,1]$, there is a quantum algorithm that uses only local (weak) measurements of the projectors $\Pi_f$ and terminates with success with probability at least $1-\epsilon$, 
		performing at most $\tilde{\mathcal{O}}\left(\frac{R\cdot|F|\cdot(|F|+R\cdot d)}{\gamma^2}\cdot\log\left(\frac{1}{\epsilon}\right)
		+\frac{|F|}{\gamma^F}\cdot\log\left(\frac{1}{\delta}\right)\cdot\log\left(\frac{1}{\epsilon}\right)\right)$ measurements in total, where $\gamma$ is the uniform gap and $\gamma^F$ is the gap of $H$ as in Definition~\ref{def:uniform_gap}. Conditional on termination with ``SUCCESS", the output quantum state is $\delta$-close in trace distance to a quantum state supported on the ground space of $H$.
	\end{corollary}	
	\begin{proof}
		We boost the previous theorem using standard techniques. After performing the algorithm described in Theorem~\ref{thm:approxRun} we get result ``SUCCESS" with probability $1/2$, conditioned on this outcome we run Algorithm~\ref{alg:proj} with $\tau=\frac{|F|}{\gamma^F}\cdot\ln\left(\frac{8}{\delta}\right)$. Since the output state was $1/2$ close to a ground state, a projective measurement of the kernel of $H$ finds the state 
		in the kernel with probability at least $1/2$, so Algorithm~\ref{alg:proj} finds a flaw with probability at most $1/2$. Moreover conditioned on not finding a flaw the state becomes $\delta$-close to a ground state due to Lemma~\ref{lemma:projPrec} and a straightforward triangle inequality argument.
		If we repeat the whole procedure $4\cdot\ln\left(\frac{1}{\epsilon}\right)$ times, 
		then the probability of all runs failing is less than $\epsilon$.
	\end{proof}
	\begin{corollary} \label{cor:SLCRuntime}
		\noindent For the case of \eqref{eq:SLC} the runtime bound of Corollary~\ref{cor:boosted} is $$\tilde{\mathcal{O}}\left(\frac{|F|^3\cdot n^2}{d\cdot\gamma^2}
		\cdot\log\left(\frac{1}{\epsilon}\right)+\frac{|F|}{\gamma^F}
		\cdot\log\left(\frac{1}{\delta}\right)\log\left(\frac{1}{\epsilon}\right)\right).$$
	\end{corollary}
	\begin{proof}
		Use that in Theorem~\ref{thm:generalR} in the case of \eqref{eq:SLC} we get $R=\mathcal{O}\left(\frac{n\cdot |F|}{d}\right)$.
	\end{proof}	

\section{Discussion}\label{sec:discussion}
	In the non-commuting case, it is computationally hard to distinguish between a state which has $0$-energy, and a state which has an exponentially small energy. Therefore, it seems unlikely that an efficient algorithm will be able to construct exactly a $0$-energy state. The energy levels in the classical and commuting case are integers, and so the energy gap was never an issue in previous works. 
	
	The main question that this work leaves open is whether there exists a polynomial time randomized approximation scheme (FPRAS) for this problem: given an energy level $\epsilon$ find, with a constant probability, a state which is supported on energy levels below $\epsilon$, with the runtime scaling polynomially in the input parameters and $\frac{1}{\epsilon}$. Appendix~\ref{apx:fpras} shows how a simple extension of our framework can be adapted to the FPRAS requirements, and a surprising phenomenon that poses a barrier for showing an FPRAS: Suppose a state $\ket{\psi}$ has energy $E$, for $H_1+H_2$ where both terms are positive semi-definite. Now, suppose the energy of the state is measured only with respect to $H_1$. It could be that the energy will increase above $E$, with high probability (even though we subtracted $H_2$, and therefore the expected energy cannot increase). 

\ignore{
\section{Quantum supremacy}
	Our quantum algorithm uses the quantum resources quite efficiently to attack a genuinely quantum problem. Also the structure of the algorithm is sort of self correcting, providing a kind of built in error-correction. These properties make it a good candidate for quantum supremacy demonstration.
	
	Actually the algorithm solves the following problem. Given a set of projectors satisfying Shearer's condition, it concludes one of the two:
	\begin{itemize}
		\item Either it prepares (and therefore possibly samples from) a quantum state with energy below epsilon.
		\item Or it concludes, that the uniform gap of the Hamiltonian is smaller than epsilon. 
	\end{itemize}
	It is possible that the algorithm can do both, but it can always solve at least one of the problems.
	Both things seems hard on a classical computer, however more investigation is needed.
}

\subsection*{Author Contributions} A.G. is the principal author of this paper. 
	A.G. adapted the ideas of \cite{HarveyVondrak15} to the commuting quantum setting.
	O.S. showed that a modified version of the resulting algorithm terminates in the non-commuting case.
	A.G. generalised the measurement procedure to fit the non-commuting setting and proposed the use of weak measurements and the quantum Zeno effect for efficient implementation, and analysed the resulting algorithms.

\subsection*{Acknowledgments}
	A.G. thanks Ronald de Wolf for support and many valuable discussions, 
	Márió Szegedy for recommending relevant literature on the classical constructive LLL, 
	Martin Schwarz and Niel de Beaudrap for discussions.
	A.G. was supported by ERC Grant 615307-QPROGRESS.
	O.S. was supported by ERC Grant 030-8301. 
	
\ifpdf
	\typeout{^^J *** PDF mode *** } 
	\printbibliography	
\else
	\typeout{^^J *** DVI mode ***} 
	\bibliography{LLL}{}
	\bibliographystyle{acmUrlePrint}
\fi 	
	
\appendix
\ignore{	
    \section{Nomenclature}\label{sec:nomenclature}
	\begin{nopreview}
	\printnomenclature
	\end{nopreview}
}	

	\section{Bounds on weighted stable set sequences}\label{apx:SSSBounds}
	\begin{lemma} \label{lemma:pathEsts}
		If the vector of probabilities $(p_f)_{f\in F}\in(0,1)^{|F|}$ satisfies
		\begin{align*}
		\text{\textbullet } \text{ \eqref{eq:ALC}},\,\,
		&\text{then } \sum_{{\cal{I}}\in{\cal{IS}}}p_{\cal I}\leq \prod_{f\in F}\frac{1}{1-x_f}.\\
		\text{\textbullet } \text{ \eqref{eq:CEC}},\,\,
		&\text{then } \sum_{{\cal{I}}\in{\cal{IS}}}p_{\cal I}\leq \prod_{f\in F}(1+y_f).\\
		\text{\textbullet } \text{ \eqref{eq:SHC}},\,\,
		&\text{then } \sum_{{\cal{I}}\in{\cal{IS}}}p_{\cal I}\leq \frac{1}{q_\varnothing(p_f)}.
		\end{align*}
	\end{lemma}
	\begin{proof}
		The corresponding proofs can be found in \cite[Lemma 5.7]{HarveyVondrak15,HarveyVondrak15Ar}, \cite[Theorem 5.10]{HarveyVondrak15} and \cite[Corollary 5.28]{HarveyVondrak15Ar}, respectively.
	\end{proof}
	
	The following theorem gives an exponential tail bound on the weighted sum of independent set sequences, if we assume the probabilities satisfy the corresponding criterion with $\epsilon$-slack.
	
	\begin{theorem}\label{thm:slack}
		Let $\epsilon \in \mathbb{R}_+$, and $p'_f=p_f\cdot(1+\epsilon)$. If $(p'_f)_{f\in F}\in(0,1)^{|F|}$ satisfies the condition
		\begin{align*}
		&\text{\textbullet } \text{ \eqref{eq:ALC}, and }
		T=-\sum_{f\in F}\ln\left(1-x_f\right)/\ln(1+\epsilon), \text{ or}\\
		&\text{\textbullet } \text{ \eqref{eq:CEC}, and }
		T=\phantom{-}\sum_{f\in F}\ln\left(1+y_f\right)/\ln(1+\epsilon), \text{ or}\\
		&\text{\textbullet } \text{ \eqref{eq:SHC}, and }
		T=\phantom{\sum_{f\in F}}-\ln\left(q_\varnothing(p'_f)\right)/\ln(1+\epsilon), \\
		\end{align*}
		then
		$$
		\sum_{{\cal I}\in{\cal IS}_{\geq(T+r)}} p_{\cal I} \leq (1+\epsilon)^{-r}.
		$$
	\end{theorem}
	\begin{proof}
		\begin{align*}
		\sum_{{\cal I}\in{\cal IS}_{\geq(T+r)}}
		&\leq \sum_{{\cal I}\in{\cal IS}_{\geq(T+r)}}
		\prod_{I\in{\cal I}}\prod_{f\in I}((1+\epsilon)p_f) \cdot (1+\epsilon)^{-(T+r)}\\
		&\leq (1+\epsilon)^{-(T+r)} \cdot \sum_{{\cal I}\in{\cal IS}} 
		\prod_{I\in{\cal I}}\prod_{f\in I}p'_f \\
		&= (1+\epsilon)^{-(T+r)} \cdot  \sum_{{\cal I}\in{\cal IS}} p'_{\cal I} \\
		&\leq (1+\epsilon)^{-r}
		\end{align*}
		where in the last step we used Lemma~\ref{lemma:pathEsts}.
	\end{proof}

	The next theorem is an improved version of the above, and uses results of \cite{HarveyVondrak15Ar} showing that there is a considerable slack naturally appearing.
	
	For ease of notation let us introduce the shorthand $q_I:=q_I(p_f)$.
	
	\begin{theorem}\label{thm:TailWithoutSlack}
		If $(p_f)_{f\in F}\in(0,1)^{|F|}$ satisfies the condition
		\begin{align*}
		&\text{\textbullet } \text{ \eqref{eq:ALC}, and }
		T=4\cdot\!\left(\sum_{f\in F} \frac{x_f}{1-x_f}\right)\!\cdot\!
		\left(t+1+\min\!\left[\ln\left(\frac{1}{q_{\varnothing}}\right),\sum_{f\in F}\ln\left(\frac{1}{1-x_f}\right)\right]\right)
		,\text{ or}   \\
		&\text{\textbullet } \text{ \eqref{eq:CEC}, and }
		T=4\cdot\!\left(\sum_{f\in F} y_f\right)\!\cdot\!
		\left(t+1+\min\!\left[\ln\left(\frac{1}{q_{\varnothing}}\right),\sum_{f\in F}\ln\left(1+y_f\right)\right]\right)
		,\text{ or}   \\
		&\text{\textbullet } \text{ \eqref{eq:SHC}, and }
		T=4\cdot\!\left(\sum_{f\in F} \frac{q_{\{f\}}}{q_{\varnothing}}\right)\!\cdot\!
		\left(t+1+\min\!\left[\ln\left(\frac{1}{q_{\varnothing}}\right),\sum_{f\in F}\ln\left(1+\frac{q_{\{f\}}}{q_{\varnothing}}\right)\right]\right),  
		\end{align*} 
		then
		$$
		\sum_{{\cal I}\in{\cal IS}_{\geq T}} p_{\cal I} \leq e^{-t},
		$$			
	\end{theorem}
	\begin{proof} In the following proof it will be convenient to prove the statements in reversed order:
		\\ \noindent\textbullet\,\,\eqref{eq:SHC}: 
		Let $\epsilon=\frac{q_\varnothing}{2\cdot \sum_{f\in F}q_{\{f\}}}$, and $(p'_f)_{f\in F}=(1+\epsilon)\cdot (p_f)_{f\in F}$. Then $(p'_f)_{f\in F}$ satisfies \eqref{eq:SHC} and $q_\varnothing(p'_f)\geq q_\varnothing/2$ as shown in \cite[Lemma 5.33]{HarveyVondrak15Ar}. 
		As shown in \cite[Claim 5.23]{HarveyVondrak15Ar} 
		$\frac{1}{q_\varnothing}\leq \prod_{f\in F}\left(1+\frac{q_{\{f\}}}{q_\varnothing}\right)$,
		and thus $\frac{1}{q_\varnothing(p'_f)}\leq2\cdot\min\left[\frac{1}{q_\varnothing},\prod_{f\in F}\left(1+\frac{q_{\{f\}}}{q_\varnothing}\right)\right]$.
		
		Let $r=\ln(1+\epsilon)\cdot t$, then applying 
		Theorem~\ref{thm:slack} completes the proof using the additional observation that $\epsilon\in(0,1]$ 
		and thus $\frac{1}{\ln(1+\epsilon)}\leq \frac{2}{\epsilon}$.
		\\ \noindent\textbullet\,\,\eqref{eq:CEC}: 	
		Due to Proposition~\ref{prop:ShearerImplied} \eqref{eq:SHC} holds,
		moreover $\forall f\in F: \frac{q_{\{f\}}(p_f)}{q_{\varnothing}(p_f)}\leq y_f$
		as shown in \linebreak\cite[Corollary 5.43]{HarveyVondrak15Ar}. Substituting this formula into the \eqref{eq:SHC} case gives the required result.$\kern-4mm$	
		\\ \noindent\textbullet\,\,\eqref{eq:ALC}: 	
		Due to Proposition~\ref{prop:ShearerImplied} \eqref{eq:SHC} holds,
		moreover $\forall f\in F: \frac{q_{\{f\}}(p_f)}{q_{\varnothing}(p_f)}\leq \frac{x_f}{1-x_f}$
		as shown in \linebreak\cite[Corollary 5.38]{HarveyVondrak15Ar}. Substituting this formula into the \eqref{eq:SHC} case gives the required result.$\kern-4mm$	
	\end{proof}
	\begin{corollary}\label{cor:expected} Let 
		\begin{equation*}
		R = \sum_{k=1}^{\infty}\min\left(1,\sum_{{\cal I}\in{\cal IS}_k} p_{\cal I}\right).
		\end{equation*}
		
		If the vector of probabilities $(p_f)_{f\in F}\in(0,1)^{|F|}$ satisfies the condition
		\begin{align*}
		&\text{\textbullet } \text{ \eqref{eq:ALC}, then } 
		R \leq
		1+4\cdot\left(\sum_{f\in F} \frac{x_f}{1-x_f}\right)\cdot
		\left(1+\min\left[\ln\left(\frac{1}{q_{\varnothing}}\right),\sum_{f\in F}\ln\left(\frac{1}{1-x_f}\right)\right]\right).  \\
		&\text{\textbullet } \text{ \eqref{eq:CEC}, then } 
		R \leq
		1+4\cdot\left(\sum_{f\in F} y_f\right)\cdot
		\left(1+\min\left[\ln\left(\frac{1}{q_{\varnothing}}\right),\sum_{f\in F}\ln\left(1+y_f\right)\right]\right).  \\
		&\text{\textbullet } \text{ \eqref{eq:SHC}, then } 
		R \leq
		1+4\cdot\left(\sum_{f\in F} \frac{q_{\{f\}}}{q_{\varnothing}}\right)\cdot
		\left(1+\min\left[\ln\left(\frac{1}{q_{\varnothing}}\right),\sum_{f\in F}\ln\left(1+\frac{q_{\{f\}}}{q_{\varnothing}}\right)\right]\right). 	  
		\end{align*}	
		\ignore{		
			\begin{align*}
			&\text{\textbullet } \text{ \eqref{eq:ALC}, is at most } 
			1+4\cdot\left(\sum_{f\in F} \frac{x_f}{1-x_f}\right)\cdot
			\left(2+\min\sum_{f\in F}\ln\left(\frac{1}{1-x_f}\right)\right).  \\
			&\text{\textbullet } \text{ \eqref{eq:CEC}, is at most } 
			1+4\cdot\left(\sum_{f\in F} y_f\right)\cdot
			\left(2+\sum_{f\in F}\ln\left(1+y_f\right)\right).  \\
			&\text{\textbullet } \text{ \eqref{eq:SHC}, is at most } 
			1+4\cdot\left(\sum_{f\in F} \frac{q_{\{f\}}(p_f)}{q_{\varnothing}(p_f)}\right)\cdot
			\left(2+\sum_{f\in F}\ln\left(1+\frac{q_{\{f\}}(p_f)}{q_{\varnothing}(p_f)}\right)\right).  
			\end{align*} }
	\end{corollary}
	\begin{proof}
		Choose $t=0$ and let $T$ be the number we get from Theorem~\ref{thm:TailWithoutSlack}, then 
		\begin{align*}
		\sum_{k=1}^{\infty}\min\left(1,\sum_{{\cal I}\in{\cal IS}_k} p_{\cal I}\right) 
		\leq T+\sum_{k=\lceil T\rceil}^{\infty}\min\left(1,\sum_{{\cal I}\in{\cal IS}_k} p_{\cal I}\right) 	
		\leq T+\sum_{{\cal I}\in{\cal IS}_{\geq \lceil T\rceil}} p_{\cal I}	
		\leq T+e^{-t} = T+1.
		\end{align*}
	\end{proof}
	 
	\section{Ambiguity in singular value decomposition}\label{apx:uniqueMap}
	In this appendix we show that for any square matrix $A=W\Sigma U^\dagger$ the linear map $W\sgn{\Sigma}U^\dagger$ is well defined, i.e., it is independent of the choice of unitaries in the singular value decomposition. 
	
	Let $A\in\mathbb{C}^{k\times k}$, $\sigma_1\geq \sigma_2\geq \ldots\geq \sigma_k$ its singular values with multiplicity, and $\Sigma=\text{Diag}(\sigma_{1\ldots n})$. 
	Suppose $A=W_1\Sigma U_1^\dagger$ and $A=W_2\Sigma U_2^\dagger$ are two singular value decompositions, with 
	$X^\dagger_i=X_i^{-1}$ for $X\in\{W,U\},i\in\{1,2\}$.
	Then $W_1\sgn{\Sigma}U_1^\dagger=W_2 \sgn{\Sigma}U_2^\dagger$, as the following shows:
	\begin{align*}
	U_1\Sigma^2 U_1^\dagger=A^\dagger &A=U_2\Sigma^2 U_2^\dagger\\
	\Downarrow& \text{}\\
	U_2^\dagger U_1\Sigma^2=&\,\Sigma^2 U_2^\dagger U_1\\
	\Updownarrow& \\
	\left[U_2^\dagger U_1,\Sigma^2\right]=&\,0\\	
	\Updownarrow& \, (0\preceq\Sigma)\\
	\left[U_2^\dagger U_1,\Sigma\right]=&\,0\left(=\left[U_2^\dagger U_1,\sgn{\Sigma}\right]\right)\\
	\Downarrow& \text{}\\
	W_1\Sigma U_1^\dagger=W_2\Sigma U_2^\dagger=W_2\Sigma U_2^\dagger U_1 U_1^\dagger
	=&\,W_2 U_2^\dagger U_1\Sigma U_1^\dagger\\
	\Downarrow& \text{}\\
	W_1\Sigma=&\,W_2 U_2^\dagger U_1\Sigma\\	
	\Updownarrow& \text{}\\
	W_1\sgn{\Sigma}=&\,W_2 U_2^\dagger U_1\sgn{\Sigma}\\
	\Downarrow& \text{}\\
	W_1\sgn{\Sigma}U_1^\dagger=&\,W_2 U_2^\dagger U_1\sgn{\Sigma}U_1^\dagger
	=W_2 \sgn{\Sigma}U_2^\dagger U_1 U_1^\dagger=W_2 \sgn{\Sigma}U_2^\dagger\\	
	\end{align*}
	\anote{Uniqueness of singular value decomposition: http://math.stackexchange.com/questions/644327/how-unique-on-non-unique-are-u-and-v-in-singular-value-decomposition-svd}	
	 
    \section{Decoherence}\label{apx:Decoherence}
	In this appendix we illustrate how decoherence arises from our weak measurement procedure (Algorithm~\ref{alg:measurement}).
    We study the ideal limiting case where the weakness parameter of the measurement $\theta$ is infinitesimally small,
    and we do infinitely many repetitions of the weak and strong measurements. Our argument is hand-wavy but can be made precise by taking the appropriate limits.
    
    Let $\Pi_f \Pi_V=W\Sigma U^\dagger$, $\sigma_i=\Sigma_{ii}$, and  $w_i=W_{.i}$, $u_i=U_{.i}$ be the $i$-th column of $W$ and $U$ respectively.
    Let $\ket{\psi^t}$ denote the unnormalised state after $t$ weak and strong measurements, 
    corresponding to the case when no positive $\Pi_f$ neither negative $\Pi_V$ measurement outcomes were observed, and let
    $\ket{\psi^t}=\sum_{j}a_j^t \ket{u_j}$, where $a_j^t$ is the amplitude of $\ket{u_j}$ in $\ket{\psi^t}$.
    The argument described in the introduction shows, that $\ket{\psi^{t+1}}\approx\Pi_V\ket{\psi^t}-\theta/2\Pi_V\Pi_f\ket{\psi^t}$.
    Assuming that $\Pi_V\ket{\psi^0}=\ket{\psi^0}$, we get
    $\sum_{j}a_j^{t+1} \ket{u_j}=\ket{\psi^{t+1}}\approx\ket{\psi^t}-\theta/2\Pi_V\Pi_f\Pi_V\ket{\psi^t}
    =\sum_{j}\left(1-\frac{\theta}{2}\sigma_j^2\right)a_j^t\ket{u_j}$. For small $\theta$ we can move to a continuous time approximation, and use the differential equation $\dot{a}_j\approx -\frac{\theta}{2}\sigma_j^2 a_j$, which yields the solution $a^t_j\approx e^{-\frac{\theta}{2}\sigma_j^2t}a^0_j$.
    
    Let $\rho^t$ denote the unnormalised density operator corresponding to cases when a positive $\Pi_f$ measurement outcome was observed in the $t$-th iteration, and therefore Algorithm~\ref{alg:measurement} terminated.
    Then $\rho^t=\sqrt{\theta}\Pi_f\ket{\psi^t}\bra{\psi^t}\Pi_f\sqrt{\theta}
    =\theta\Pi_f\Pi_V\ket{\psi^t}\bra{\psi^t}\Pi_V\Pi_f
    =\theta W\Sigma U^\dagger\ket{\psi^t}\bra{\psi^t}U\Sigma W^\dagger$.
    Then $\rho^{t+1}_{ij}:=\bra{w_i}\rho^{t+1}\ket{w_j}
    =\theta\left(\sigma_ia^{t}_i\right)\cdot\left(\sigma_ja^{t}_j\right)^*
    \approx\theta\sigma_i\sigma_je^{-\frac{\theta}{2}\left(\sigma_i^2+\sigma_j^2\right)t}
     a^{0}_i\cdot\left(a^{0}_j\right)^*$. 
    We can approximate $\rho_{ij}^{\text{out}}:=\sum_{t=1}^{\infty}\rho^t_{ij}
    \approx \int_{0}^\infty\theta\sigma_i\sigma_je^{-\frac{\theta}{2}\left(\sigma_i^2+\sigma_j^2\right)t}
    a^{0}_i\cdot\left(a^{0}_j\right)^*\, dt=\frac{2\sigma_i\sigma_j}{\sigma_i^2+\sigma_j^2}\rho^{\text{in}}_{ij}$, 
    where we defined $\rho^{\text{in}}_{ij}=\bra{u_i}\rho^{\text{in}}\ket{u_j}$ 
    with $\rho^{\text{in}}=\ket{\psi^0}\bra{\psi^0}$. 
    The change of basis $u_i\rightarrow w_j$ in $\rho^{\text{in}}\rightarrow\rho^{\text{out}}$ corresponds to the unitary map $WU^\dagger$, which we described as the exact quantum channel.
    
    This little calculation also explains, that for infinitesimally small $\theta$ the procedure is always successful, and projects out the complete overlap with $\Pi_f$ if repeated indefinitely. Also it is converging exponentially to its infinite version. 
	The strength of the decoherence depends on the difference between the singular values, and does not happen at all if the singular values equal, since then $\frac{2\sigma_i\sigma_j}{\sigma_i^2+\sigma_j^2}=1$. Note that this phenomenon is only present for non-commuting projectors, since in the commuting case $\sigma_i\in\{0,1\}$.
    \onote{For next version, consider adding an example about decoherence, which is commented below}
    
    
    
    
    
	\section{Jordan's Theorem on two orthogonal projectors}\label{apx:Jordan}
	
	In order to get insight to the interplay between two orthogonal projectors we refer to the work of Camille Jordan \cite{cJordan1875} from 1875. The structure theorem we invoke also plays an important role in many other results form quantum computation,
	often in an implicit way, see e.g. \cite{MarriottWatrous,SzegedyQuantumMarkov}. A modern treatment of the following theorem together with a proof can be found in \cite[Theorem VII.1.8]{BhatiaMatrixAnal97}:
	
	\begin{theorem} \label{thm:twoProjectors}
		For any two orthogonal projectors $\Pi_f,\Pi_V$ acting on the Hilbert space $\mathcal{H}$ there exists an orthogonal decomposition of $\mathcal{H}$ into one-dimensional and two-dimensional subspaces that are invariant under both $\Pi_f$ and $\Pi_V$. Moreover,
		inside each two-dimensional subspace, $\Pi_f$ and $\Pi_V$ are rank-one projectors (in other words, inside each two-dimensional subspace there are two unit vectors $\ket{w_i}$ and $\ket{u_i}$ such that $\Pi_f$ projects on $\ket{w_i}$ and $\Pi_V$ projects on	$\ket{u_i}$), and $|\braket{w_i}{u_i}|\notin \{0,1\}$.
	\end{theorem}
	
	The theorem above gives a lot of insight to our algorithms, and explains why and how the singular values govern the behaviour of Algorithm~\ref{alg:measurement} -- compare Appendix~\ref{apx:Decoherence}. For example it shows that the unitary operation $WU^\dagger$ in Definition~\ref{def:exactChannel} basically performs a rotation inside each two-dimensional subspace, while sends $\ket{u_i}$ to $\ket{w_i}$. We can also use this theorem to give a more insightful proof of Proposition~\ref{prop:subspaceIdentities}. 
	
	Note that in the commuting case, when $\Pi_V$ and $\Pi_f$ commute, they can be diagonalized simultaneously, implying that all the subspaces are 1 dimensional, and therefore $WU^\dagger=\Id$. 
	\onote{Explain in more details the relation to $U$, and $\Pi_V - \Pi_{V_f}$.}

	\section{A note on the commutativity of resampling operations}\label{apx:commExampla}
	A notion of commutativity plays a crucial role in the analysis of generalised Moser-Tardos algorithms in the unifying work of Kolmogorov~\cite{KolmogorovComm}. The commutativity of two resampling operations basically refers to the case when resampling two independent flaws in different orders gives the same result. We will now give an example showing  that this can fail to hold even if all projectors commute. Thus in this appendix we assume that all projectors commute.
	
	In the quantum setting, an individual resampling operation on flaw $f$ should consist of replacing qubits $b(f)$ by maximally mixed ones. This step would itself commute if applied on non-adjacent flaws, but this operation may also be followed by doing measurements on adjacent flaws. This measurement step is necessary to actually keep track of which flaws are present. So one step would be $\displaystyle\operatornamewithlimits{M}_{\Gamma^+(f)}\circ R_f$, where $\displaystyle\operatornamewithlimits{M}_{S}$ denotes the operation of measuring flaw $\Pi_f$ for all $f\in S$, and $\circ$ denotes composition.
	Since all projectors commute we may assume knowing all present flaws of $\rho$. Therefore, $\rho=\displaystyle\operatornamewithlimits{M}_{F}(\rho)$. 
	Since resampling qubits adjacent to $f$ does not affect flaws that are in $F\setminus\Gamma^+(f)$, 
	we may perform a non-destructive measurement on them, i.e.,
	$\displaystyle\operatornamewithlimits{M}_{F}\circ R_f(\rho)=\displaystyle\operatornamewithlimits{M}_{\Gamma^+(f)}\circ R_f(\rho)$
	if $\rho=\displaystyle\operatornamewithlimits{M}_{F}(\rho)$. 
	
	So the question we ask is whether for every set of commuting projectors $\{\Pi_f : f\in F\}$ and for all $a, b \in F$
	\begin{align}\label{eqn:resComm}
	\operatornamewithlimits{M}_{F}\circ R_{b}\circ 
	\operatornamewithlimits{M}_{F}\circ R_{a}\circ 
	\operatornamewithlimits{M}_{F}(\rho) \overset{?}{=}
	\operatornamewithlimits{M}_{F}\circ R_{a}\circ 
	\operatornamewithlimits{M}_{F}\circ R_{b}\circ 
	\operatornamewithlimits{M}_{F}(\rho) ?
	\end{align}
	
	Let $F=\{a,b,c\}$ and $\Pi^{loc}_a=\ket{1}\bra{1}_1\otimes \Id_{\{2\}}$, $\Pi^{loc}_b=\Id_{\{3\}}\otimes\ket{1}\bra{1}_4$ and 
	$\Pi^{loc}_c=(\frac{3}{5}\ket{00}_{2,3}+\frac{4}{5}\ket{11}_{2,3})\cdot(\frac{3}{5}\bra{00}_{2,3}+\frac{4}{5}\bra{11}_{2,3})$ 
	be projectors acting on qubits $b(a)=\{1,2\}$, $b(b)=\{3,4\}$ and $b(c)=\{2,3\}$ correspondingly (the numbers in subscripts denote corresponding qubits $1-4$). Then with $\rho = \ket{1111}\bra{1111}$ equation~\eqref{eqn:resComm} does not hold.
	Note that we cheat a bit here as we could have defined $b(a)=\{1\}$, $b(b)=\{4\}$, in which case the operations would naturally commute. However a basis change on qubits $\{1,2\}$ and $\{3,4\}$ could result in an entangled projector justifying $b(a)=\{1,2\}$ and $b(b)=\{3,4\}$, however it would probably also require changing $b(c)$ to $\{1,2,3,4\}=b(a)\cup b(b)$. 
	
	We mention as a curiosity that when we tried disproving equation~\eqref{eqn:resComm} using Kitaev's toric code~\cite{KitaevToric} interestingly our simulations showed commutation in the resampling operations.
	
	\section{An example justifying our quantum channel definitions} \label{apx:WhyRotate}
	Throughout this appendix we are going to use $\circ$ for composition. Also for a projector $\Pi$ we interpret $\Pi\circ\rho$ as $\Pi\rho\Pi$. 

	In this appendix, we will give a ``correct'' algorithm, in the sense that it respects the loop invariant, and therefore successfully finds a ground state upon termination. The point of this algorithm is to explain why it is crucial to resample the qubits only after $\Pi_f$ is violated, and not only when the state has been found to have some overlap with $\Pi_f$ (by measuring $\Pi^{C\cap \{f\}}$).
	\begin{algorithm}[H]
		\caption{An alternative algorithm}\label{alg:alter}
		\begin{algorithmic}[1]
			\STATE {\bf input} constraints $\{\Pi_f\}_{f\in F}$
			\STATE set all qubits to the maximally mixed state, and mark all constraints as unchecked.
			\STATE {\bf while} there is a $\Pi_f$ which is unchecked {\bf do}
			
			\STATE ~~~~ measure $\Pi^{C \cup f}$ 
			\STATE ~~~~  {\bf if} the measurement was violated {\bf then}
			\STATE ~~~~~~~~  resample all qubits of $\Pi_f$, i.e., replace them by uniformly random qubits
            \STATE ~~~~~~~~  mark all constraints in $\Gamma^+(\Pi_f)$ as unchecked
            \STATE ~~~~ {\bf else}
            \STATE ~~~~~~~~ mark $\Pi_f$ as checked

			\STATE ~~~~  {\bf end if}
			\STATE {\bf end while}
			\STATE {\bf terminate with} ``SUCCESS" 		
		\end{algorithmic}
	\end{algorithm}
	The induction hypothesis for the loop invariant is that in the entrance to the while loop $\Pi^C \circ \rho = \rho$. We show that the induction hypothesis remains true after failed measurement, which is the interesting case: 
    \begin{align*}\Pi^{C \setminus \Gamma^+(f)} \circ R_{b(f)} \circ \left(\Id - \Pi^{C\cup f}\right) \circ  \rho &=  R_{b(f)} \circ \Pi^{C \setminus \Gamma^+(f)} \circ \left(\Pi^C - \Pi^{C\cup f}\right)  \circ \rho \\
    &=R_{b(f)} \circ \left(\Pi^{C} - \Pi^{C\cup f}\right) \circ \rho \\
    &=R_{b(f)} \circ \left(\Id - \Pi^{C\cup f}\right) \circ \rho 
    \end{align*}    
    In the first equation we used the fact that the resampling the qubits $b(f)$, and measuring $\Pi^{C \setminus \Gamma(f)}$ commute, since they act on different qubits; and the induction hypothesis.
    
    In the following example, the parameters that govern the number of resamplings in our main algorithm -- see Theorem~\ref{thm:generalR} -- are kept fixed, yet the number of resamplings in the suggested algorithm above are unbounded.

		There are only 2 qubits, and two projectors. $ \Pi_1=\ket{\psi}\bra{\psi}$
		acts on two qubits where $\ket{\psi} = \sqrt \epsilon \ket{00} +
		\sqrt{1-\epsilon}\ket{11}$ for some small $\epsilon$. $\Pi_2=\ket{1}\bra{1}$ acts
		only on the second qubit. The reader should verify that the only satisfying state of both projectors is $\ket{10}$, and therefore $\Pi^{1,2}=\ket{10}\bra{10}$.
		
		Suppose we start with  $\ket{\psi_{init}}$ in the state $\ket{00}$ or $\ket{01}$
		(which happens with probability $1/2$).
		When we test $\Pi^{1}$ the answer will almost always be ``checked", and
		the state might change a little bit if we started with $\ket{00}$ and won't
		change at all if we started with $\ket{01}$. When we test $\Pi^{1,2}$ the
		outcome will almost always be ``no", and the second qubit will be resampled.
		But this will \emph{not} help, as we (almost) go back to one of the initial states: we
		are very close to either $\ket{00}$ or $\ket{01}$ (and very far from the only accepted
		state $\ket{10}$).

	\section{Failed attempt towards an FPRAS}\label{apx:fpras}

	It seems fairly straightforward to adapt the progressive channel formalism to the setting of a quantum fully polynomial random approximation scheme (FPRAS); by this we mean the construction of a quantum algorithm, which for any given $\epsilon$, should find a state $\rho$ with support on energy levels below $\epsilon$, but can terminate with ``error" with probability $\frac{1}{2}$ (which can be reduced exponentially by repetition). The running time should grow polynomially in the input size and in $1/\epsilon$.  The challenge is to implement at the $t$'th iteration, a progressive quantum channel as in Definition~\ref{def:progressive_channel}, where we replace $\Pi_{V^C}$ with $\Pi_{V^C}^ {t \epsilon}$ which is the projection onto all eigenstates with energy up to $ t \epsilon$ with respect to $\sum_{f\in C} \Pi_f$, and similarly replace $\Pi_{V^{C\cup \{f\}}}$ with $\Pi_{V^{C\cup \{f\}}}^{(t+1)\epsilon}$, and replace $\Pi_{V^{C\setminus\Gamma^+(f)}}$ with $\Pi_{V^{C\setminus\Gamma^+(f)}}^{(t+1)\epsilon}$. This guarantees that if the number of channels applied in the algorithm is $T$, and the algorithm terminates successfully, the energy of the final state is below $T \epsilon$. 
	
	The main issue can be understood via the following example, which arises with the use of (weak or strong) measurements: Suppose $\ket{\psi}$ lies in $\Pi_{V^C}^ {t \epsilon}$. Now, suppose we do a (strong) measurement of $\Pi_f$, then we are tempted to think that $\Pi_f\ket{\psi}$ lies in $\Pi_{V^{C\setminus\Gamma^+(f)}}^ {t \epsilon}$. Surprisingly this need not to be true, as it can happen that $\ket{\psi}$ is not an element of the subspace $\Pi_{V^{C\setminus\Gamma^+(f)}}^ {t \epsilon}$.
	
	Next, we will show how another attempt to implement such a channel fails. Consider the adaptation to the exact quantum channel as in Definition~\ref{def:exactChannel}, where we use $\Pi_{V^{C\cup \{f\}}}^{(t+1)\epsilon}$ and $\Pi_{V^C}^ {(t+1) \epsilon}$ instead of their corresponding original definitions.
	
	This construction would not have property~\ref{it:Abad}. Essentially, the reason is that unlike before, $\Pi^\epsilon_V \Pi^\epsilon_{V \cup \{f \}} \neq  \Pi^\epsilon_{V \cup \{f \}}$. Therefore, it could be that a state has low energy with respect to a Hamiltonian, and that when the energy of that state is measured of with respect to a sub-Hamiltonian, its energy would be higher than before.

\end{document}

%% file: myBiblatex.tex
\usepackage[autostyle]{csquotes}

\usepackage[backend=bibtex,bibencoding=ascii,maxnames=10,doi=false,url=false,eprint=true,isbn=false,citestyle=numeric-comp,sorting=nyt,firstinits=true]{biblatex}

\makeatletter
\def\blx@maxline{77}
\makeatother

\AtEveryBibitem{%
	\clearlist{language}%
}
		
\renewbibmacro*{volume+number+eid}{%
	\mkbibbold{\printfield{volume}}%
	\printfield{number}%
	\setunit{\addcomma\space}%
	\printfield{eid}}
\DeclareFieldFormat[article]{number}{\mkbibparens{#1}}	

\renewbibmacro*{in:}{%
	\ifentrytype{article}{}{\printtext{\bibstring{in}\intitlepunct}}}

\renewbibmacro*{journal}{%
	\iffieldundef{shortjournal}
	{\iffieldundef{journaltitle}
		{}
		{\printtext[journaltitle]{%
				\printfield[titlecase]{journaltitle}%
				\setunit{\subtitlepunct}%
				\printfield[titlecase]{journalsubtitle}}}}
	{\printfield{shortjournal}}}

\newbibmacro{string+doiurl}[1]{%
	\iffieldundef{doi}{%
		\iffieldundef{url}{%
			#1%
		}{%
		\href{\thefield{url}}{#1}%
	}%
}{%
\href{http://dx.doi.org/\thefield{doi}}{#1}%
}%
}

\DeclareFieldFormat{title}{\usebibmacro{string+doiurl}{\mkbibemph{#1}}}
\DeclareFieldFormat[thesis]{title}{\usebibmacro{string+doiurl}{\mkbibemph{#1}}}
\DeclareFieldFormat[article,incollection,inproceedings,unpublished]{title}{\usebibmacro{string+doiurl}{\mkbibemph{#1}}}
